\documentclass[twocolumn,aps,pra,showkeys,longbibliography]{revtex4-1}
\usepackage[breaklinks=true]{hyperref}
\hypersetup{
    colorlinks,
    linkcolor={red!50!black},
    citecolor={blue!50!black},
    urlcolor={blue!80!black}
}
\usepackage{blindtext}
\usepackage{amsmath,amssymb,amsthm,mathtools,mathrsfs}
\usepackage{tikz}
\usepackage{pifont}
\usepackage{comment}
\usepackage{adjustbox}
\usepackage{dsfont}
\usepackage[normalem]{ulem}
\usepackage{enumerate}
\usepackage[all]{xy} 
\usetikzlibrary{circuits.ee.IEC}

\theoremstyle{plain}
\newtheorem{theorem}{Theorem}[section]
\newtheorem{lemma}[theorem]{Lemma}
\newtheorem{proposition}[theorem]{Proposition}
\newtheorem{corollary}[theorem]{Corollary}
\theoremstyle{definition}
\newtheorem{definition}[theorem]{Definition}
\newtheorem{remark}[theorem]{Remark}

\newcommand\define[1]{\emph{\textbf{#1}}}

\newcommand{\bd}{\begin{definition}}
\newcommand{\ed}{\end{definition}}
\newcommand{\bt}{\begin{theorem}}
\newcommand{\et}{\end{theorem}}
\newcommand{\bc}{\begin{corollary}}
\newcommand{\ec}{\end{corollary}}
\newcommand{\blem}{\begin{lemma}}
\newcommand{\elem}{\end{lemma}}
\newcommand{\be}{\begin{equation}}
\newcommand{\ee}{\end{equation}}
\newcommand{\bn}{\begin{proposition}}
\newcommand{\en}{\end{proposition}}
\newcommand{\bq}{\begin{question}}
\newcommand{\eq}{\end{question}}
\newcommand{\bprf}{\begin{proof}}
\newcommand{\eprf}{\end{proof}}
\newcommand{\br}{\begin{remark}}
\newcommand{\er}{\end{remark}}

\let\C=\Chi

\newcommand{\matr}{\mathbb{M}}

\def\R{{{\mathbb R}}}
\def\C{{{\mathbb C}}}

\newcommand{\id}{\mathrm{id}}

\newcommand{\Tr}{{\rm Tr} }

\newcommand{\<}{\langle}
\renewcommand{\>}{\rangle}

\newcommand{\Ad}{\mathrm{Ad}}

\newcommand{\Choi}{\mathscr{C}}
\newcommand{\Jamio}{\mathscr{J}}
\newcommand{\SWAP}{{\scriptstyle\mathrm{SWAP}}}

\DeclareFontFamily{OT1}{pzc}{}
\DeclareFontShape{OT1}{pzc}{m}{it}{ <-> s*[1.2] pzcmi7t }{}
\DeclareMathAlphabet{\mathpzc}{OT1}{pzc}{m}{it}
\newcommand{\Alg}[1]{\mathpzc{#1}}

\newlength\stateheight
\newlength\minimumstatewidth
\tikzset{width/.initial=\minimummorphismwidth}
\tikzset{colour/.initial=white}

\newif\ifblack\pgfkeys{/tikz/black/.is if=black}
\newif\ifwedge\pgfkeys{/tikz/wedge/.is if=wedge}
\newif\ifvflip\pgfkeys{/tikz/vflip/.is if=vflip}
\newif\ifhflip\pgfkeys{/tikz/hflip/.is if=hflip}
\newif\ifhvflip\pgfkeys{/tikz/hvflip/.is if=hvflip}

\pgfdeclarelayer{foreground}
\pgfdeclarelayer{background}
\pgfsetlayers{background,main,foreground}

\def\thickness{0.4pt}

\makeatletter
\pgfkeys{%
  /tikz/on layer/.code={
    \pgfonlayer{#1}\begingroup
    \aftergroup\endpgfonlayer
    \aftergroup\endgroup
  },
  /tikz/node on layer/.code={
    \gdef\node@@on@layer{%
      \setbox\tikz@tempbox=\hbox\bgroup\pgfonlayer{#1}\unhbox\tikz@tempbox\endpgfonlayer\pgfsetlinewidth{\thickness}\egroup}
    \aftergroup\node@on@layer
  },
  /tikz/end node on layer/.code={
    \endpgfonlayer\endgroup\endgroup
  }
}
\def\node@on@layer{\aftergroup\node@@on@layer}
\makeatother

\makeatletter
\pgfdeclareshape{state}
{
    \savedanchor\centerpoint
    {
        \pgf@x=0pt
        \pgf@y=0pt
    }
    \anchor{center}{\centerpoint}
    \anchorborder{\centerpoint}
    \saveddimen\overallwidth
    {
        \pgf@x=3\wd\pgfnodeparttextbox
        \ifdim\pgf@x<\minimumstatewidth
            \pgf@x=\minimumstatewidth
        \fi
    }
    \savedanchor{\upperrightcorner}
    {
        \pgf@x=.5\wd\pgfnodeparttextbox
        \pgf@y=.5\ht\pgfnodeparttextbox
        \advance\pgf@y by -.5\dp\pgfnodeparttextbox
    }
    \anchor{A}
    {
        \pgf@x=-\overallwidth
        \divide\pgf@x by 4
        \pgf@y=0pt
    }
    \anchor{B}
    {
        \pgf@x=\overallwidth
        \divide\pgf@x by 4
        \pgf@y=0pt
    }
    \anchor{text}
    {
        \upperrightcorner
        \pgf@x=-\pgf@x
        \ifhflip
            \pgf@y=-\pgf@y
            \advance\pgf@y by 0.4\stateheight
        \else
            \pgf@y=-\pgf@y
            \advance\pgf@y by -0.4\stateheight
        \fi
    }
    \backgroundpath
    {
       \begin{pgfonlayer}{foreground}
        \pgfsetstrokecolor{black}
        \pgfsetlinewidth{\thickness}
        \pgfpathmoveto{\pgfpoint{-0.5*\overallwidth}{0}}
        \pgfpathlineto{\pgfpoint{0.5*\overallwidth}{0}}
        \ifhflip
            \pgfpathlineto{\pgfpoint{0}{\stateheight}}
        \else
            \pgfpathlineto{\pgfpoint{0}{-\stateheight}}
        \fi
        \pgfpathclose
        \ifblack
            \pgfsetfillcolor{black!50}
            \pgfusepath{fill,stroke}
        \else
            \pgfusepath{stroke}
        \fi
       \end{pgfonlayer}
    }
}

\pgfdeclareshape{my ground}
{
  \inheritsavedanchors[from=rectangle ee]
  \inheritanchor[from=rectangle ee]{center}
  \inheritanchor[from=rectangle ee]{north}
  \inheritanchor[from=rectangle ee]{south}
  \inheritanchor[from=rectangle ee]{east}
  \inheritanchor[from=rectangle ee]{west}
  \inheritanchor[from=rectangle ee]{north east}
  \inheritanchor[from=rectangle ee]{north west}
  \inheritanchor[from=rectangle ee]{south east}
  \inheritanchor[from=rectangle ee]{south west}
  \inheritanchor[from=rectangle ee]{input}
  \inheritanchor[from=rectangle ee]{output}
  \anchorborder{\csname pgf@anchor@my ground@input\endcsname}

  \backgroundpath{
    \pgf@process{\pgfpointadd{\southwest}{\pgfpoint{\pgfkeysvalueof{/pgf/outer xsep}}{\pgfkeysvalueof{/pgf/outer ysep}}}}
    \pgf@xa=\pgf@x \pgf@ya=\pgf@y
    \pgf@process{\pgfpointadd{\northeast}{\pgfpointscale{-1}{\pgfpoint{\pgfkeysvalueof{/pgf/outer xsep}}{\pgfkeysvalueof{/pgf/outer ysep}}}}}
    \pgf@xb=\pgf@x \pgf@yb=\pgf@y
    \pgfpathmoveto{\pgfqpoint{\pgf@xa}{\pgf@ya}}
    \pgfpathlineto{\pgfqpoint{\pgf@xa}{\pgf@yb}}
    \pgfmathsetlength\pgf@xa{.5\pgf@xa+.5\pgf@xb}
    \pgfmathsetlength\pgf@yc{.16666\pgf@yb-.16666\pgf@ya}
    \advance\pgf@ya by\pgf@yc
    \advance\pgf@yb by-\pgf@yc
    \pgfpathmoveto{\pgfqpoint{\pgf@xa}{\pgf@ya}}
    \pgfpathlineto{\pgfqpoint{\pgf@xa}{\pgf@yb}}
    \advance\pgf@ya by\pgf@yc
    \advance\pgf@yb by-\pgf@yc
    \pgfpathmoveto{\pgfqpoint{\pgf@xb}{\pgf@ya}}
    \pgfpathlineto{\pgfqpoint{\pgf@xb}{\pgf@yb}}
  }
}
\makeatother

\tikzset{inline text/.style =
  {text height=1.2ex,text depth=0.25ex,yshift=0.5mm}}
\tikzset{arrow box/.style =
  {rectangle,inline text,fill=white,draw,
    minimum height=5mm,yshift=-0.5mm,minimum width=5mm}}
\tikzset{bubble/.style =
  {inner sep=0mm,minimum width=3mm,minimum height=3mm,
    draw,shape=circle,fill=white}}
\tikzset{dot/.style =
  {inner sep=0mm,minimum width=1mm,minimum height=1mm,
    draw,shape=circle}}
\tikzset{white dot/.style = {dot,fill=white,text depth=-0.2mm}}
\tikzset{scalar/.style = {diamond,draw,inner sep=1pt}}
\tikzset{square/.style =
  {inner sep=0mm,minimum width=2mm,minimum height=2mm,
    draw,shape=rectangle}}

\tikzset{star/.style = {dot,fill=white,text depth=-0.2mm}}
\tikzset{copier/.style = {dot,fill,text depth=-0.2mm}}
\tikzset{fakecopier/.style = {square,fill,text depth=-0.2mm}}
\tikzset{discarder/.style = {my ground,draw,inner sep=0pt,
    minimum width=4.2pt,minimum height=11.2pt,anchor=input,rotate=90}}

\tikzset{xshiftu/.style = {shift = {(#1, 0)}}}
\tikzset{yshiftu/.style = {shift = {(0, #1)}}}
\tikzset{scriptstyle/.style={font=\everymath\expandafter{\the\everymath\scriptstyle}}}

\begin{document}
\title{Bipartite quantum states admitting a causal explanation}
\author{Minjeong Song}
\email{song.at.qit@gmail.com}
\affiliation{Centre for Quantum Technologies, National University of Singapore}
\author{Arthur J.~Parzygnat}
\email{arthurjp@mit.edu}
\affiliation{Experimental Study Group, Massachusetts Institute of Technology, Cambridge, Massachusetts 02139, USA}

\begin{abstract}
The statistics of local measurements of joint quantum systems can sometimes be used to distinguish the spatiotemporal structure in which they were measured. We first prove that every bipartite separable density matrix is temporally compatible with direct causal influence for arbitrary finite-dimensional quantum systems and measurements of a  tomographically-complete class of observables, which includes all Pauli observables in the case of multi-qubit systems. Equivalently, if a bipartite density matrix is not temporally compatible with direct causal influence, then it must be entangled. We also provide an operational meaning for the two temporal evolutions consistent with such correlations in terms of generalized dephasing channels and pretty good measurements. The two temporal evolutions are Bayesian inverses of each other, which is different from them being Petz recovery maps of each other. Finally, we prove necessary and sufficient conditions for an arbitrary bipartite quantum state to be temporally compatible, thereby providing a temporal analogue of the positive partial transpose criterion valid for quantum systems of any dimension. 
\end{abstract}

\keywords{Temporal compatibility, causal compatibility, quantum state over time, quantum conditioning, pseudo-density matrix, Schur channel, Hadamard channel, generalized dephasing channel, pretty good measurement, Petz recovery map, quantum Bayes' rule, light-touch observable, Jordan product}

\maketitle

\section{Introduction}

The causal compatibility problem asks what possible causal structures are compatible with a given multipartite probability distribution~\cite{navascues2020inflation,wolfe2021qinflation,WolfeSpekkensFritz2019,FrKl23,PearlCausality2009,ABHLS17,CostaShrapnel16,FrKiSpWo2025,FrKi2024,Pienaar2020,BaLoOr19}. While a bipartite distribution $\mathbb{P}(x,y)$ always admits a potential causal explanation by conditioning via Bayes' rule~\cite{Ba1763,ChJa18,Fr20}
\begin{equation}
\mathbb{P}(y|x)\mathbb{P}(x)=\mathbb{P}(x,y)=\mathbb{P}(x|y)\mathbb{P}(y),
\end{equation}
an analogous statement does not hold for quantum systems~\cite{RAVJSR15}. The quantum causal compatibility problem for two observers can be stated as follows~\cite{FJV15,SNREG23,liu2023quantum,JSK23}. Alice and Bob choose a complete set of Pauli observables $\{M_{a}\}$ and $\{N_{b}\}$ for Hilbert spaces $\Alg{H}_{\Alg{A}}$ and $\Alg{H}_{\Alg{B}}$, respectively, representing qubit or multiqubit systems. They repeatedly perform such measurements on their respective systems after reinitializing the systems in order to gather statistics of the product of their measurement outcomes. These statistics asymptotically lead to a collection of expectation values $\langle M_a,N_b \rangle$ indexed by $a$ and $b$. The \emph{causal compatibility problem} then asks what causal structure is compatible with the expectation values $\langle M_a,N_b \rangle$. 

To formalize the causal compatibility problem more mathematically, Ref.~\cite{FJV15} showed that such Pauli correlations $\langle M_a,N_b \rangle$ can be used to define a \emph{pseudo-density matrix} $R$, an analogue of a density matrix that incorporates spatiotemporal quantum correlations, by the expression
\be
\label{eqn:PDMgeneral}
R := \sum_{a,b} \frac{\langle M_a,N_b \rangle}{\dim(\Alg{H}_{\Alg{A}})\dim(\Alg{H}_{\Alg{B}})}  M_a\otimes N_b.
\ee
Conversely, we can obtain these expectation values $\langle M_a,N_b \rangle$ from $R$ by using the fact that the Pauli observables form an orthogonal basis with respect to the Hilbert--Schmidt inner product. Namely, 
\be
\label{eqn:representabilityofEVs}
\Tr\big[R(M_a \otimes N_b)\big] = \langle M_a,N_b \rangle
\ee
for all $a,b$. 

For example, imagine that Alice and Bob make independent measurements without direct causal influence, in which case $R$ is a density matrix in ${\Alg{A}\otimes\Alg{B}}$. This is the case, for example, when Alice and Bob are spatially separated. In this case, we say that the expectation values and associated pseudo-density matrix are \emph{spatially compatible}. As another example, one could imagine the temporal structure with direct causal influence in which first Alice has an initial state $\rho_{\Alg{A}}\in \Alg{A}$ and performs a measurement of $M_{a}$, which induces a post-measurement state as defined by L{\"u}ders' rule~\cite{busch2009luders,BLM1996,hegerfeldt2012luders,khrennikov2009luders,Lu06,vN18,FiWe25}. This post-measurement state then evolves under a quantum channel $\Alg{A}\xrightarrow{\mathcal{E}}\Alg{B}$ only to be retrieved by Bob, who then performs a measurement of $N_{b}$. In this case, the joint expectation value is given by 
\be
\label{eq:twotimeexpectvalues}
\langle M_a,N_b \rangle=\sum_{x,y} xy\, \Tr\Big[ \mathcal{E}(M_{x|a} \rho_{\Alg{A}} M_{x|a}) N_{y|b} \Big],
\ee 
where $\{M_{x|a}\}$ and $\{N_{y|b}\}$ represent the spectral projection operators associated with the Pauli observables~\cite{BMKG13,BCL90,Johansen2007,Wigner1963}. For instance, $M_{x|a}$ signifies that Alice measures observable $M_{a}$ and obtains outcome $x$, which is an eigenvalue of $M_{a}$. One then obtains a pseudo-density matrix from these expectation values via~\eqref{eqn:PDMgeneral}. Similarly, one obtains an analogous pseudo-density matrix if the temporal order is reversed, i.e., Bob makes the measurement first, the system evolves, and then Alice measures. In either of these two cases, we say that the expectation values and the associated pseudo-density matrices are \emph{temporally compatible}.

But now, imagine being given only the expectation values $\<M_{a},N_{b}\>$ without being told which of these three causal structures are compatible with $\<M_{a},N_{b}\>$. Can we then passively certify what the underlying causal structure was purely by inspection of the \emph{observed} statistics? Said differently, with which statistics can we \emph{not} distinguish a spatial structure and a temporal structure? 

Working directly with the expectation values $\<M_{a},N_{b}\>$, and hence the pseudo-density matrix $R$ from~\eqref{eqn:PDMgeneral}, checking spatial compatibility involves verifying whether $R$ is positive or not, since this is exactly when $R$ defines a density matrix (it is assumed that the marginal statistics are compatible with ordinary density matrices so that the trace of $R$ is guaranteed to be $1$). On the other hand, checking temporal compatibility requires more careful inspection. Fortunately, the latter can be obtained explicitly, thanks to the observations made in Refs.~\cite{HHPBS17,liu2023quantum,FuPa24a} that the pseudo-density matrix for the expectation values in~\eqref{eq:twotimeexpectvalues} is given by 
$R=\frac{1}{2}\big\{\rho\otimes1_{\Alg{B}},\Jamio[\mathcal{E}]\big\}$,
where $\{C,D\}=CD+DC$ is the Jordan product and $\Jamio[\mathcal{E}]$ is the Jamio{\l}kowski matrix of $\mathcal{E}$~\cite{Ja72}. Using this expression, we provide a condition for causal compatibility in Theorem~\ref{thm:TempCompat} by viewing 
$R$ as an example of a \emph{state over time} associated with the initial state $\rho$ and quantum channel $\mathcal{E}$~\cite{FuPa22a,HHPBS17,FuPa22,FuPa24a}. 

It is crucial to note that a pseudo-density matrix could be compatible with multiple causal structures. In such a case, we cannot distinguish one causal structure from another in principle using these passive methods, thus giving a fundamental limitation on inferring quantum causal structures. That said, there still remains a quantum advantage for inferring quantum causal structure as compared to the analogous situation in classical probability and statistics~\cite{RAVJSR15}, where if one causal structure is compatible, then Bayes' rule implies that the other two causal structures are also compatible.
These two points illustrate that it is important to identify when we can and cannot distinguish these causal structures. 

A simple example of expectation values that are both spatially and temporally compatible comes from product states (i.e., statistically independent quantum systems) by virtue of there being no correlation. However, it was left as a conjecture in Ref.~\cite{SNREG23} that every separable state, i.e., a statistical mixture of product states, might also be temporally compatible. In other words, do the set of pseudo-density matrices that are both spatially and temporally compatible contain the set of separable states? Moreover, which mixed quantum states (possibly entangled) are both spatially and temporally compatible~\cite{RAVJSR15}? How can we characterize the set of correlations that are both spatially and temporally compatible?

The present paper characterizes this overlapping set of joint measurement statistics when these statistics are obtained from measuring a tomographically-complete set of observables, which we term \emph{light-touch observables}~\cite{liu2023quantum,FuPa24a,PaFu24TSC}. Such observables consist of binary-outcome observables together with the identity operator, and they extend Pauli observables for qubit systems to arbitrary qudit systems, i.e., quantum systems of arbitrary finite dimension. Although restricting to such observables might seem like a limitation, note that causal inference is a difficult problem, even classically. Therefore, it is important to isolate easily computable methods, even at the cost of generality, in order to push the boundaries of our knowledge on quantum causal inference. 

To start, we prove in Theorem~\ref{thm:tempcompSB} that the conjecture from Ref.~\cite{SNREG23}, namely that separable states are temporally compatible, is true (cf.\ Figure~\ref{fig:venndiagram}). 
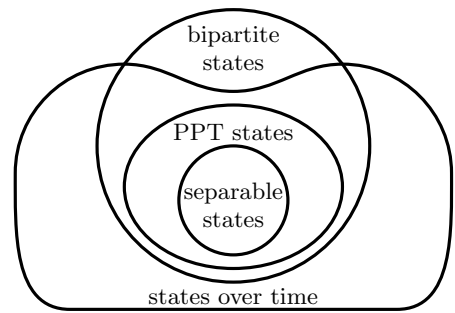
\begin{figure}
\begin{tikzpicture}[scale=0.725]
\draw[very thick] (0,0) circle (2.5cm);
\node at (0,1.625) {\begin{tabular}{c}bipartite\\states\end{tabular}};
\draw[very thick] (0,-1.375) circle (1.125cm);
\node at (0,-1.375) {\begin{tabular}{c}separable\\states\end{tabular}};
\draw[very thick] (0,-1.00) ellipse (2cm and 1.5cm);
\node at (0,0) {PPT states};
\draw[very thick] (0,-3.0) -- (3.0,-3.0) to[out=0,in=-90] (4.0,-0.5) to[out=90,in=0] (2.0,1.5) to[out=180,in=0] (0,1.00) to[out=180,in=0] (-2.0,1.5) to[out=180,in=90] (-4.0,-0.5) to[out=-90,in=180] (-3.0,-3.0) -- cycle;
\node at (0,-2.75) {states over time};
\end{tikzpicture}
\caption{A Venn diagram depicting the sets of bipartite states, separable states, PPT states, and states over time, the last of which contains temporally compatible pseudo-density matrices (although the sets of bipartite states, separable states, and PPT states are convex, the set of states over time is not convex~\cite{SNREG23}). The containment of separable and PPT states inside states over time are proved in Theorem~\ref{thm:tempcompSB} and Corollary~\ref{cor:PPTimpliestempcompat}, respectively.
}
\label{fig:venndiagram}
\end{figure}
In Theorem~\ref{thm:splittemporalmap}, we provide an operational interpretation of the channel $\mathcal{E}$ from~\eqref{eq:twotimeexpectvalues} in terms of a generalized dephasing channel followed by a measure-and-prepare protocol implemented by a pretty good measurement~\cite{hausladen1994pretty}. This provides a potential mechanism exhibiting the causal structure. Afterwards, we generalize our results further and find necessary and sufficient conditions for \emph{any} bipartite state to be temporally compatible in Theorem~\ref{thm:TempCompat}. Interestingly, this necessary and sufficient condition is closely related to the positive partial transpose criterion~\cite{choi1982positive,peres1996separability,horodecki1996necessary,MaTy13}, and yet it holds for arbitrary finite-dimensional quantum systems (cf.\ Corollary~\ref{cor:temporalPPTcriterion} and Corollary~\ref{cor:PPTimpliestempcompat}).

The organization of this paper is as follows. Section~\ref{sec:BG} sets up notation and contains background material on joint expectation values, pseudo-density matrices, states over time, light-touch observables, and the three causal structures mentioned above. Section~\ref{sec:main} contains our main results, which were outlined in the previous paragraph. Section~\ref{sec:discussion} concludes the paper with a discussion and some open questions. There are four appendices that prove the dephasing channel mentioned above is well-defined, discuss the relation to state discrimination, extend Theorem~\ref{thm:TempCompat} as a general test for temporal compatibility, and remove the assumptions on strict positivity of marginal density matrices, respectively.

\section{Background and Notation}
\label{sec:BG}

\subsection{Standard definitions and setting notation}

The following definitions can be found in Refs.~\cite{We89,Ja72,FuPa22a,NiCh11,BCRS19,Fa01,HallQuantum13}. 
Algebras of quantum systems will be denoted as $\Alg{A}$, $\Alg{B}$, \dots, and will always be taken to be full matrix algebras. Meanwhile, a classical system described by a finite set $\Theta$ will be represented by a commutative algebra of functions $\C^{\Theta}=\{f:\Theta\to\C\}$, where multiplication is defined pointwise. We note that \emph{algebras} in this work always mean finite-dimensional $C^*$-algebras, which allow us to discuss classical and quantum systems using the same language. 
A state $\tau\in\Alg{A}\otimes\Alg{B}$ is \define{separable} with respect to the factorization $\Alg{A}\otimes\Alg{B}$ iff there exists a finite set $\Theta$, density matrices $\rho_{\Alg{A};\theta}\in\Alg{A},\rho_{\Alg{B};\theta}\in\Alg{B}$, and a probability $t_{\theta}$ on $\Theta$ such that 
\begin{equation}
\label{eqn:sepstate}
\tau=\sum_{\theta\in\Theta}t_{\theta}\rho_{\Alg{A};\theta}\otimes\rho_{\Alg{B};\theta}. 
\end{equation}
Otherwise, $\tau$ is said to be \define{entangled} with respect to the factorization $\Alg{A}\otimes\Alg{B}$. 
The \define{Hilbert--Schmidt adjoint} $\Alg{B}\xrightarrow{\mathcal{E}^*}\Alg{A}$ of a linear map $\Alg{A}\xrightarrow{\mathcal{E}}\Alg{B}$ is the unique linear map satisfying
$
\Tr\big[\mathcal{E}(A)^{\dag}B\big]=\Tr\big[A^{\dag}\mathcal{E}^*(B)\big]
$
for all $A\in\Alg{A}$ and $B\in\Alg{B}$, where $C^{\dag}$ denotes the conjugate transpose (adjoint) of $C$ and $\Tr$ denotes the trace.
Classical, quantum, and hybrid channels are modeled by completely positive trace-preserving (CPTP) maps. 
Hermitian-preserving trace-preserving maps will be abbreviated as HPTP maps. 
The identity channel from $\Alg{A}$ to itself will be denoted by $\id_{\Alg{A}}$. 
The \define{Jamio{\l}kowski operator} $\Jamio[\mathcal{E}]$ associated with any linear map $\Alg{A}\xrightarrow{\mathcal{E}}\Alg{B}$ is the element of $\Alg{A}\otimes\Alg{B}$ given by 
\be
\Jamio[\mathcal{E}]:=(\id_{\Alg{A}}\otimes\mathcal{E})\big(\mu_{\Alg{A}}^{*}(1_{\Alg{A}})\big),
\ee
where $1_{\Alg{A}}$ is the unit element in $\Alg{A}$ and $\mu_{\Alg{A}}$ is the multiplication map in $\Alg{A}$, which is uniquely characterized by $\mu_{\Alg{A}}(A_1\otimes A_2)=A_1 A_2$ for all $A_1,A_2\in\Alg{A}$ and linearity. 
In the case that $\Alg{A}=\matr_{m}$, where $\matr_{m}$ is the algebra of complex $m\times m$ matrices, it is given by 
\be
\Jamio[\mathcal{E}]=\sum_{i,j}|i\>\<j|\otimes\mathcal{E}\big(|j\>\<i|\big)
=(\id_{\Alg{A}}\otimes\mathcal{E})(\SWAP),
\ee
where $\{|i\>\}$ denotes an orthonormal basis and $\SWAP$ is the swap operator in $\matr_{m}\otimes\matr_{m}$ given by $\SWAP=\sum_{i,j}|i\>\<j|\otimes|j\>\<i|$. 
The Jamio{\l}kowski matrix of $\mathcal{E}$ is not to be confused with the Choi matrix of $\mathcal{E}$, which is defined as
\begin{equation}
\Choi[\mathcal{E}]=\sum_{i,j}|i\>\<j|\otimes\mathcal{E}\big(|i\>\<j|\big)
=(\id_{\Alg{A}}\otimes\mathcal{E})\big(|\Omega\>\<\Omega|\big), 
\end{equation}
where $|\Omega\>:=\sum_{i}|i\>\otimes|i\>$ is the (unnormalized) maximally entangled state with respect to a specified basis $\{|i\>\}$~\cite{Ch75}. 

\subsection{Pseudo-density matrices and the canonical state over time}

Pseudo-density matrices were introduced in Ref.~\cite{FJV15} as extensions of density matrices to spatiotemporally-related quantum systems consisting of several qubits in order to represent spatiotemporal correlations~\cite{BMKG13,Fritz10}. Soon after, Ref.~\cite{HHPBS17} compiled a list of ``quantum states over time'' showing that not a single example from their selection of spatiotemporal extensions of density matrices satisfy a list of natural and mathematically convenient axioms analogous to what are satisfied in the classical setting. Later, Ref.~\cite{FuPa22} showed that, contrary to the no-go theorem of Ref.~\cite{HHPBS17}, a particular state over time, when the notion is defined more carefully~\cite{FuPa22,FuPa22a}, does indeed satisfy the axioms proposed in Ref.~\cite{HHPBS17}. Subsequently, several characterization theorems for this state over time appeared in Refs.~\cite{PFBC23,LiNg23,FuPa24a} (see also Refs.~\cite{BDOV13,BDOV14}), which then motivated calling it the canonical state over time. 
We briefly summarize these definitions and results in this section.

First, we define two-time expectation values, then we review pseudo-density matrices, and finally we define light-touch observables and the canonical state over time, which extends pseudo-density matrices to arbitrary finite-dimensional quantum systems. In what follows, if $\Alg{A}$ and $\Alg{B}$ are algebras, then a pair $(\mathcal{E},\rho_{\Alg{A}})$, with $\rho_{\Alg{A}}\in\Alg{A}$ a density matrix and $\Alg{A}\xrightarrow{\mathcal{E}}\Alg{B}$ a CPTP map, is called a \define{process}. The set of all such processes from $\Alg{A}$ to $\Alg{B}$ will be denoted by $\mathscr{P}(\Alg{A},\Alg{B})$.

\bd
Let $M\in \Alg{A}$ and $N\in \Alg{B}$ be observables, with $M$ having a canonical spectral decomposition $M=\sum_{i=1}^{m}\lambda_i P_i$, where $P_{i}$ is the projector onto the eigenspace of $M$ associated with eigenvalue $\lambda_i$.
Then the \define{two-time expectation value} of the observables $M$ and $N$ with respect to the process $(\mathcal{E},\rho_{\Alg{A}})\in\mathscr{P}(\Alg{A},\Alg{B})$ is the real number 
\be \label{eqn:twotimeexpectationvalues}
\<M,N\>_{(\mathcal{E},\rho_{\Alg{A}})}=\sum_{i}\lambda_i \Tr\big[\mathcal{E}(P_i\rho_{\Alg{A}} P_i)N\big] .
\ee
\ed
Note that L{\"u}ders' update rule is used in~\eqref{eqn:twotimeexpectationvalues}~\cite{busch2009luders,BLM1996,Lu06}.
In general, there does not exist an operator $\varrho_{\Alg{A}\Alg{B}}\in\Alg{A}\otimes\Alg{B}$ that \emph{represents} two-time expectation values in the sense that $\Tr[\varrho_{\Alg{A}\Alg{B}}(M\otimes N)]=\<M,N\>_{(\mathcal{E},\rho_{\Alg{A}})}$ for \emph{all} observables $M\in\Alg{A}$ and $N\in\Alg{B}$~\cite{FuPa24a}. This is primarily due to the fact that the measurement process is nonlinear in the first observable $M$ due to the state update rule. This lack of representability is in contrast to the case of expectation values on a spatially-separated joint system $\Alg{A}\otimes\Alg{B}$, where a joint density matrix $\rho_{\Alg{A}\Alg{B}}$ represents the joint expectation values of \emph{all} observables.
Despite the lack of such an operator $\varrho_{\Alg{A}\Alg{B}}$ for two-time expectation values, Ref.~\cite{FJV15} showed that such expectation values admit such a representation for systems of qubits provided that one restricts to Pauli observables rather than all observables (an explicit example will be provided in Section~\ref{sec:threecausalscenarios}). In what follows, we denote the identity operator and three Pauli operators by  
\be
\sigma_{0}=\begin{bmatrix}1&0\\0&1\end{bmatrix},
\;
\sigma_{1}=\begin{bmatrix}0&1\\1&0\end{bmatrix},
\;
\sigma_{2}=\begin{bmatrix}0&-i\\i&0\end{bmatrix},
\;
\sigma_{3}=\begin{bmatrix}1&0\\0&-1\end{bmatrix}.
\ee
We will occasionally also refer to $\sigma_{0}$ as a Pauli observable. 

\bn
\label{prop:PDM}
Let $\Alg{A}$ and $\Alg{B}$ each denote a system of $m$ qubits, i.e.,  $\Alg{A}=\Alg{B}=\matr_{2^{m}}=\matr_{2}^{\otimes m}$. Given $\alpha\in \{0,1,2,3\}^{m}$, set $\alpha_{j}\in \{0,1,2,3\}$ to be the $j^{\text{th}}$ component of $\alpha$ and $\sigma_{\alpha}\in \matr_{2}^{\otimes m}$ to be the observable given by $\sigma_{\alpha}=\sigma_{\alpha_1}\otimes \cdots \otimes \sigma_{\alpha_m}$. Then, for every process $(\mathcal{E},\rho_{\Alg{A}})\in \mathscr{P}(\Alg{A},\Alg{B})$,  there exists a unique operator $\varrho_{\Alg{A}\Alg{B}}\in\Alg{A}\otimes\Alg{B}$, called the \define{pseudo-density matrix} associated with $(\mathcal{E},\rho_{\Alg{A}})$, such that 
$\<\sigma_{\alpha},\sigma_{\beta}\>_{(\mathcal{E},\rho_{\Alg{A}})}=\Tr[\varrho_{\Alg{A}\Alg{B}} (\sigma_{\alpha}\otimes\sigma_{\beta})]$ for all $\alpha,\beta\in\{0,1,2,3\}^{m}$. Moreover, an explicit formula for $\varrho_{\Alg{A}\Alg{B}}$ is given by 
\be
\varrho_{\Alg{A}\Alg{B}}=\frac{1}{4^m}\sum_{\alpha,\beta\in \{0,1,2,3\}^{m}}\<\sigma_{\alpha},\sigma_{\beta}\>_{(\mathcal{E},\rho_{\Alg{A}})}\sigma_{\alpha}\otimes \sigma_{\beta}
\ee
\en

Refs.~\cite{HHPBS17,liu2023quantum} showed that the pseudo-density matrix in Proposition~\ref{prop:PDM} may alternatively be given by the formula 
$
\varrho_{\Alg{A}\Alg{B}}=\frac{1}{2}\{\rho_{\Alg{A}}\otimes 1_{\Alg{B}},\Jamio[\mathcal{E}]\},
$
where $\rho_{\Alg{A}}=\Tr_{\Alg{B}}[\varrho_{\Alg{A}\Alg{B}}]$ and $\Alg{A}\otimes\Alg{B}\xrightarrow{\Tr_{\Alg{B}}}\Alg{A}$ denotes the partial trace. 
Motivated by the fact that the right-hand-side agrees with the canonical state over time of Ref.~\cite{FuPa22}, Ref.~\cite{FuPa24a} extended Proposition~\ref{prop:PDM} to arbitrary finite-dimensional quantum systems by isolating some key features of Pauli operators, such as the fact that their spectrum is always $+1$ or $\pm 1$.

\bd
\label{defn:lighttouch}
For any algebra $\Alg{A}$, a \define{light-touch observable} in $\Alg{A}$ is 
an observable $A\in\Alg{A}$ such that the spectrum of $A$ is either $\{\lambda\}$ or $\{\pm\lambda\}$ for some $\lambda\ge0$. 
\ed

\bn
\label{prop:oprep}
Let $\Alg{A}$ and $\Alg{B}$ be algebras of arbitrary dimension and let $(\mathcal{E},\rho_{\Alg{A}})\in \mathscr{P}(\Alg{A},\Alg{B})$ be a process.
Then there exists a unique element 
$\varrho_{\Alg{A}\Alg{B}}\in\Alg{A}\otimes\Alg{B}$ such that 
\be
\<A,B\>_{(\mathcal{E},\rho_{\Alg{A}})}=\Tr\big[\varrho_{\Alg{A}\Alg{B}}(A\otimes B)\big]
\ee
for all light-touch observables $A\in \Alg{A}$ and for all observables $B\in\Alg{B}$. Moreover, $\varrho_{\Alg{A}\Alg{B}}$ is the \define{canonical state over time}  
$
\varrho_{\Alg{A}\Alg{B}}:=\frac{1}{2}\big\{\rho_{\Alg{A}}\otimes 1_{\Alg{B}},\Jamio[\mathcal{E}]\big\}.
$
\en

Such a formula will prove to be convenient for calculations determining temporal compatibility, which we define more precisely in the next section.

\subsection{Three causal structures for two observers}
\label{sec:threecausalscenarios}

Here, we isolate the causal structures that we will consider for two observers making local measurements on quantum systems. 
Let $\{M_a\}$ and $\{N_b\}$ be sets of light-touch observables on two algebras $\Alg{A}$ and $\Alg{B}$, respectively. We interpret $\{M_a\}$ and $\{N_b\}$ as measurements of observables performed by Alice and Bob on systems $\Alg{A}$ and $\Alg{B}$, respectively. We view such measurements by their associated projective instruments $\Alg{A}\xrightarrow{\mathcal{M}_{a}}\Alg{A}\otimes\C^{X_{a}}$ and $\Alg{B}\xrightarrow{\mathcal{N}_{b}}\Alg{B}\otimes\C^{Y_{b}}$, where $X_{a}$ and $Y_{b}$ denote the sets of outcomes labeled by $x$ and $y$, respectively~\cite{DaLe70,FuPa22a,Ozawa1984}. Namely, $\mathcal{M}_{a}$ defines a collection $\mathcal{M}_{x|a}:\Alg{A}\to\Alg{A}$ of completely positive trace nonincreasing maps given by the L{\"u}ders projection map $\mathcal{M}_{x|a}(\rho)=M_{x|a}\rho M_{x|a}$ for each $x\in X_{a}$, and similarly for $\mathcal{N}_{b}$.
The induced expectation values $\<M_a,N_b\>$ are spatially or temporally compatible or incompatible in the following sense~\cite{liu2023quantum,JSK23,LCD24,SNREG23,FuPa24a}. 
        \begin{enumerate}[(a)]
        \item There exists a density matrix $\rho_{\Alg{A}\Alg{B}}\in\Alg{A}\otimes\Alg{B}$ such that $\<M_a,N_b\>=\Tr\big[\rho_{\Alg{A}\Alg{B}}(M_a\otimes N_b)\big]$ for all $a,b$. In such a case, the expectation values are said to be \define{spatially compatible} (possibly with a \define{common cause}) on the joint system $\Alg{A}\otimes\Alg{B}$. In diagrammatic notation~\cite{CoKi17,HeVi19,FGGPS22,Pa24,FrKl23}, this structure is drawn as 
        \[
        \begin{tikzpicture}[font=\small]
\node[state] (omega) at (0,0) {\;$\rho_{\Alg{A}\Alg{B}}$\;};
\coordinate (B) at (0.25,2.00) {};
\node[discarder] (X) at (0.85,1.65) {};
\coordinate (A) at (-0.85,2.00) {};
\node[discarder] (Y) at (-0.25,1.65) {};
\node[arrow box] (e) at (-0.55,0.85) {\;$\mathcal{M}_{a}$\;};
\node[arrow box] (g) at (0.55,0.85) {\;\,$\mathcal{N}_{b}$\,\;};
\draw (omega) ++(0.55, 0) to (g);
\draw[double] (g) ++(-0.30,0.25) to (B);
\draw (g) ++(0.30,0.25) to (X);
\draw (omega) ++(-0.55, 0) to (e);
\draw[double] (e) ++(-0.30,0.25) to (A);
\draw (e) ++(0.30,0.25) to (Y);
\path[scriptstyle]
node at (-0.40,1.45) {$\Alg{A}$}
node at (1.00,1.45) {$\Alg{B}$}
node at (0.53,1.85) {$\C^{Y_{b}}$}
node at (-0.75,0.35) {$\Alg{A}$}
node at (-1.13,1.85) {$\C^{X_{a}}$}
node at (0.70,0.35) {$\Alg{B}$};
\end{tikzpicture}
        \]
        where single lines refer to quantum systems (noncommutative algebras), double lines refer to classical systems (commutative algebras), the symbol \scalebox{0.65}{$\begin{tikzpicture}\node[discarder] (d) at (0,0.3) {}; \draw (0,0) -- (d); \end{tikzpicture}$} denotes tracing out, and the diagram is read from bottom to top.
        Note that there are no assumptions on the density matrix $\rho_{\Alg{A}\Alg{B}}$, eg.\ it need not be separable, entangled, etc. If no such $\rho_{\Alg{A}\Alg{B}}$ exists, then the expectation values are said to be \define{spatially incompatible}. 
        \item There exists a density matrix $\rho_{\Alg{A}}\in\Alg{A}$ and a quantum channel $\Alg{A}\xrightarrow{\mathcal{E}}\Alg{B}$ such that $\<M_a,N_b\>$ equals the two-time expectation value associated with the process $(\mathcal{E},\rho_{\Alg{A}})$, i.e., 
        \be
        \label{eqn:twotimeAtoB}
        \<M_a,N_b\>=\sum_{x\in X_{a}}\lambda_{x}\Tr\big[\mathcal{E}(M_{x|a}\rho_{\Alg{A}}M_{x|a}) N_b\big]
        \ee
        for all $a,b$, where $M_{a}=\sum_{x}\lambda_{x}M_{x|a}$
        is the canonical spectral decomposition of the observable $M_a$, so that $M_{x|a}$ is the projector onto the eigenspace of $M_{a}$ associated with eigenvalue $\lambda_x$. In such a case, the expectation values are said to be \define{temporally compatible} with \define{temporal order} $\Alg{A}\to\Alg{B}$, i.e., $\Alg{A}$ has a \define{direct influence} on $\Alg{B}$. In diagrammatic notation, this structure is drawn as
        \[
        \begin{tikzpicture}[font=\small]
\node[state] (omega) at (0,-0.8) {\,$\rho_{\Alg{A}}$\,};
\node[arrow box] (q) at (0,-0.1) {\;\;\;\;$\mathcal{M}_{a}$\;\;\;\;};
\node[arrow box] (g) at (0.40,0.8) {$\mathcal{E}$};
\node[arrow box] (N) at (0.40,1.7) {\;\,$\mathcal{N}_{b}$\,\;};
\node[discarder] (d) at (0.65,2.45) {};
\draw (omega) to (q);
\draw (q) ++(0.40,0.25) to (g);
\draw[double] (q) ++(-0.40,0.25) to (-0.40,2.85);
\draw (g) to (N); 
\draw (N) ++(0.25,0.25) to (d);
\draw[double] (N) ++(-0.25,0.25) to (0.15,2.85);
\path[scriptstyle]
node at (0.20,-0.55) {$\Alg{A}$}
node at (0.60,0.35) {$\Alg{A}$}
node at (0.60,1.30) {$\Alg{B}$}
node at (0.80,2.25) {$\Alg{B}$}
node at (-0.10,2.75) {$\C^{Y_{b}}$}
node at (-0.72,2.75) {$\C^{X_{a}}$};
\end{tikzpicture}
        \]
        \item There exists a density matrix $\rho_{\Alg{B}}\in\Alg{B}$ and a quantum channel $\Alg{B}\xrightarrow{\mathcal{F}}\Alg{A}$ such that $\<M_a,N_b\>$ equals the two-time expectation value associated with the process $(\mathcal{F},\rho_{\Alg{B}})$, i.e., 
        \be
        \label{eqn:twotimeBtoA}
        \<M_a,N_b\>=\sum_{y\in Y_{b}}\mu_{y}\Tr\big[\mathcal{F}(N_{y|b}\rho_{\Alg{B}}N_{y|b}) M_a\big]
        \ee
        for all $a,b$, where $N_{b}=\sum_{y}\mu_{y}N_{y|b}$ is the canonical spectral decomposition of the observable $N_b$, so that $N_{y|b}$ is the projector onto the eigenspace of $N_{b}$ associated with eigenvalue $\mu_{y}$. In such a case, the expectation values are said to be \define{temporally compatible} with \define{temporal order} $\Alg{B}\to\Alg{A}$, i.e., $\Alg{B}$ has a \define{direct influence} on $\Alg{A}$. In diagrammatic notation, this structure is drawn as
        \[
        \begin{tikzpicture}[font=\small]
\node[state] (omega) at (0,-0.8) {\,$\rho_{\Alg{B}}$\,};
\node[arrow box] (q) at (0,-0.1) {\;\;\;\;$\mathcal{N}_{b}$\;\;\;\;};
\node[arrow box] (g) at (0.42,0.8) {$\mathcal{F}$};
\node[arrow box] (N) at (0.42,1.7) {\;\,$\mathcal{M}_{a}$\,\;};
\node[discarder] (d) at (0.69,2.45) {};
\draw (omega) to (q);
\draw (q) ++(0.42,0.25) to (g);
\draw[double] (q) ++(-0.42,0.25) to (-0.42,2.85);
\draw (g) to (N); 
\draw (N) ++(0.27,0.25) to (d);
\draw[double] (N) ++(-0.25,0.25) to (0.17,2.85);
\path[scriptstyle]
node at (0.20,-0.55) {$\Alg{B}$}
node at (0.60,0.35) {$\Alg{B}$}
node at (0.60,1.30) {$\Alg{A}$}
node at (0.84,2.25) {$\Alg{A}$}
node at (-0.10,2.75) {$\C^{X_{a}}$}
node at (-0.70,2.75) {$\C^{Y_{b}}$};
\end{tikzpicture}
        \]
        \end{enumerate}
The first causal structure is referred to as \define{spatial}, while the latter two causal structures are referred to as \define{temporal}. 
When there does \emph{not} exist either a channel $\Alg{A}\xrightarrow{\mathcal{E}}\Alg{B}$ or a channel $\Alg{B}\xrightarrow{\mathcal{F}}\Alg{A}$ satisfying~\eqref{eqn:twotimeAtoB} or~\eqref{eqn:twotimeBtoA}, the expectation values are said to be \define{temporally incompatible} (called \emph{atemporal} in Ref.~\cite{SNREG23}).
Although other causal structures can be considered, such as those in Ref.~\cite{liu2023quantum}, we will focus on these three for simplicity. Additionally, it is important to stress that our notion of temporal compatibility is not a property of the states and evolutions alone, but also of measurements, which in this case are of light-touch observables with their corresponding projective instruments. We will comment more on this point in the Discussion section. 

\begin{remark}
\label{rmk:conditionals}
Given a bipartite density matrix $\rho_{\Alg{A}\Alg{B}}$, if the expectation values $\<M_{a},N_{b}\>$ are also temporally compatible in the $\Alg{A}\to\Alg{B}$ direction, then we call the channel $\Alg{A}\xrightarrow{\mathcal{E}}\Alg{B}$ a \define{temporal channel} associated with $\rho_{\Alg{A}\Alg{B}}$ (and similarly for $\mathcal{F}$ if the direction is $\Alg{B}\to\Alg{A}$).
One could justify calling temporal channels \define{conditionals} due to the fact that they provide analogues of conditional probabilities~\cite{PaQPL21,PaBayes,Pa17}. In fact, such conditionals indeed satisfy a quantum generalization of Bayes' rule, as shown in Refs.~\cite{FuPa22a,PaFu24TSC}. 
\end{remark}

To analyze the three causal structures above, we will utilize the fact that the associated expectation values can be represented by a matrix trace formula for all three causal structures. 
Intuitively, the representability means that the expectation values can be calculated directly from the observables and a single matrix valid for those observables, without having to know the causal structure. 
In the following definition, a collection of observables in an algebra is said to be \define{tomographically complete} iff it forms a spanning set of the algebra, where one takes complex linear combinations to define the span. 

\bd
\label{defn:representability}
Let $\Alg{A}$ and $\Alg{B}$ be two algebras. Let $\{M_a\}$ and $\{N_b\}$ be tomographically-complete sets of observables in $\Alg{A}$ and $\Alg{B}$, respectively. An associated collection of expectation values $\<M_a,N_b\>\in\R$ is \define{representable} for the observables $\{M_a\}$ and $\{N_b\}$ iff there exists a matrix $R\in\Alg{A}\otimes\Alg{B}$ such that 
\be
\label{eqn:representability}
\Tr\big[R(M_a\otimes N_b)\big]=\<M_a,N_b\>
\ee
for all $a,b$. In such a case, we say that the expectation values $\<M_{a},N_{b}\>$ are \define{represented} by $R$. 
\ed

As an immediate example, we know that if the expectation values are spatially compatible, then they are representable by a joint density matrix $\rho_{\Alg{A}\Alg{B}}$ for \emph{all} observables in $\Alg{A}$ and $\Alg{B}$.
However, expectation values coming from direct causal influence are not representable for \emph{all} observables in general. This is why the above definition of representability is made for general tomographically-complete \emph{sets} of observables, but not necessarily for their real linear combinations, and hence not necessarily \emph{all} observables. 
A simple example illustrating this point is given in Example 2.11 (and generalized in Theorem 3.2) in Ref.~\cite{FuPa24a}. Namely, let $\Alg{A}=\Alg{B}=\matr_{2}$, and let $(\mathcal{E},\rho_{\Alg{A}})$ be the process given by $\mathcal{E}=\id_{\Alg{A}}$ and $\rho_{\Alg{A}}=|-\>\<-|=\frac{1}{2}\left[\begin{smallmatrix}1&-1\\-1&1\end{smallmatrix}\right]$. Then the two-time expectation values $\<M_{a},N_{b}\>$ given by~\eqref{eqn:twotimeAtoB} are \emph{not} linear in the left coordinate $M_{a}$, while the expression on the left of~\eqref{eqn:representability} \emph{is} linear in $M_{a}$, thus leading to a contradiction (an explicit pair of observables illustrating lack of linearity is presented in Ref.~\cite{FuPa24a}).

Note, however, that when the sets of observables in Definition~\ref{defn:representability} are taken to be Pauli observables for systems of qubits (and \emph{not} their linear combinations), then the pseudo-density matrix represents their expectation values~\cite{FJV15}. A similar statement can be made for arbitrary quantum systems with light-touch observables~\cite{FuPa24a}. 
Namely, Proposition~\ref{prop:oprep} says that if the observables are light-touch observables and if the expectation values are temporally compatible, then they are represented by the hermitian matrix 
\be
\label{eq:Estarrho}
\mathcal{E}\star\rho_{\Alg{A}}:=\frac{1}{2}\big\{\rho_{\Alg{A}}\otimes 1_{\Alg{B}},\Jamio[\mathcal{E}]\big\}
\ee
in the case that $\Alg{A}$ has a direct influence on $\Alg{B}$ and 
\be
\label{eq:Binverserep}
\gamma\left(\mathcal{F}\star\rho_{\Alg{B}}\right)
=\frac{1}{2}\big\{1_{\Alg{A}}\otimes\rho_{\Alg{B}},\Jamio[\mathcal{F}^*]\big\}
\ee
in the case that $\Alg{B}$ has a direct influence on $\Alg{A}$, where $\gamma$ is the swap map defined by unique linear extension of 
\be
\gamma(B\otimes A):=A\otimes B.
\ee
Equation~\eqref{eq:Binverserep} follows from Lemma~2 in Appendix~A of Ref.~\cite{FuPa22a}. 
The representability of these three causal structures for light-touch observables will allow us to prove interesting relations between separability, entanglement, and temporal compatibility, as we show in the next section.

\section{Spatial versus temporal correlations}
\label{sec:main}

\subsection{Separability implies temporal compatibility}

\bt
\label{thm:tempcompSB}
Let $\tau$ be a separable density matrix in $\Alg{A}\otimes\Alg{B}$ and set $\rho_{\Alg{A}}=\Tr_{\Alg{B}}[\tau]$. Assume $\rho_{\Alg{A}}$ is a faithful state, i.e., $\rho_{\Alg{A}}>0$. Then there exists a unique CPTP map $\Alg{A}\xrightarrow{\mathcal{E}}\Alg{B}$ such that $\tau=\mathcal{E}\star\rho_{\Alg{A}}$. Similarly, if $\rho_{\Alg{B}}:=\Tr_{\Alg{A}}[\tau]>0$, then there exists a unique CPTP map $\Alg{B}\xrightarrow{\mathcal{F}}\Alg{A}$ such that $\tau=\gamma(\mathcal{F}\star\rho_{\Alg{B}})$.
\et

Theorem~\ref{thm:tempcompSB} says that the expectation values associated with separable density matrices in $\Alg{A}\otimes\Alg{B}$ are always temporally compatible, thereby resolving an open question in Ref.~\cite{SNREG23}. The assumption $\rho_{\Alg{A}}>0$ is not needed in this theorem and $\rho_{\Alg{A}}\ge0$ is sufficient to guarantee existence, but not necessarily uniqueness (see Appendix~\ref{sec:nonfaithful} for details). 
To prove Theorem~\ref{thm:tempcompSB}, we will make use of the following lemmas and proposition. 

\blem
\label{lem:invsumprob}
Let $\{p_{i}\}_{i=1,\dots,m}$ be a collection of real strictly positive numbers. Then the $m\times m$ matrix 
\be
\sum_{i,j}\frac{|i\>\<j|}{p_{i}+p_{j}}
\ee
is a positive matrix. 
\elem

\bprf[Proof of Lemma~\ref{lem:invsumprob}]
Such a matrix is called a Cauchy matrix, which is known to be positive (see Exercise 1.1.2 in Ref.~\cite{Bh07} for example). 
\eprf

For the next lemma, we recall that the \define{Hadamard--Schur product} $\odot$ on square matrices of the same dimension is given by 
\be 
\<i|A\odot B|j\>=\<i|A|j\>\<i|B|j\>.
\ee

\blem
\label{lem:Schurproductthm}
The Hadamard--Schur product of two positive matrices is positive. 
\elem

\bprf[Proof of Lemma~\ref{lem:Schurproductthm}]
This is called the Schur Product Theorem (see Exercise 1.2.5 on page 8 in Ref.~\cite{Bh07}). 
\eprf

\bn
\label{prop:HPTPconditional}
Let $\Alg{A}=\matr_{m}$ and $\Alg{B}=\matr_{n}$.
Let $\tau$ be any hermitian trace-1 matrix in $\Alg{A}\otimes\Alg{B}$ and set $\rho_{\Alg{A}}=\Tr_{\Alg{B}}[\tau]$. Assume $\rho_{\Alg{A}}$ is a faithful state, i.e., $\rho_{\Alg{A}}>0$. Then there exists a unique HPTP map $\Alg{A}\xrightarrow{\mathcal{E}}\Alg{B}$ such that $\tau=\mathcal{E}\star\rho_{\Alg{A}}$. Moreover, an explicit formula for $\mathcal{E}$ is given by the unique linear extension of the map that satisfies 
\be
\label{eq:qconditionalmap}
\mathcal{E}\big(|i\>\<j|\big)=\frac{2}{p_i+p_j}\Tr_{\Alg{A}}\Big[\tau\big(|i\>\<j|\otimes1_{\Alg{B}}\big)\Big]
\ee
for all $i,j\in\{1,\dots,m\}$, where $\{|i\>\}$ is an orthonormal basis of eigenvectors for $\rho_{\Alg{A}}$, i.e., $\rho_{\Alg{A}}=\sum_{i=1}^{m}p_{i}|i\>\<i|$, with $\{p_{i}\}$ a probability distribution. 
\en

We note that the map $\mathcal{E}$ obtained from linearly extending~\eqref{eq:qconditionalmap} is well-defined, i.e., it is independent of the choice of basis $\{|i\>\}$ for $\rho_{\Alg{A}}$. This will be explained in the proof of Proposition~\ref{prop:HPTPconditional} momentarily. 
The more general case where $\rho_{\Alg{A}}\ge0$ is discussed in Appendix~\ref{sec:nonfaithful}.

\bprf[Proof of Proposition~\ref{prop:HPTPconditional}]
Let us check that the formula~\eqref{eq:qconditionalmap} for $\mathcal{E}$ provides the unique solution to $\tau=\mathcal{E}\star\rho_{\Alg{A}}$ in the case that $\tau$ is an arbitrary hermitian matrix with trace equal to $1$. If such a map $\mathcal{E}$ were to exist, it must necessarily satisfy 
\be
\label{eqn:tauJordanproduct}
\tau=\mathcal{E}\star\rho_{\Alg{A}}
=\frac{1}{2}\big\{\rho_{\Alg{A}}\otimes1_{\Alg{B}},\Jamio[\mathcal{E}]\big\}.
\ee
This equation is of the form $B=AX+XA$, which is a special case of the Sylvester (sometimes Lyapunov) equation~\cite{BhRo97}. When $A$ is positive definite, it has a unique solution $X$ given by $X=\int_{0}^{\infty}e^{-t A}Be^{-t A}\,dt$. Hence, $\Jamio[\mathcal{E}]$ can be expressed as
\be
\label{eqn:JEfromSylvester}
\Jamio[\mathcal{E}]=\int_{0}^{\infty}\left(e^{-\frac{t}{2}\rho_{\Alg{A}}}\otimes1_{\Alg{B}}\right)\tau\left(e^{-\frac{t}{2}\rho_{\Alg{A}}}\otimes1_{\Alg{B}}\right)dt.
\ee
Combining this with the fact that $\Jamio$ is an isomorphism, there exists a unique (well-defined) linear map $\mathcal{E}$ satisfying~\eqref{eqn:tauJordanproduct}. 
By choosing an orthonormal eigenbasis $\{|i\>\}$ of $\rho_{\Alg{A}}$, we can obtain an alternative expression for $\mathcal{E}$. Namely, rewriting~\eqref{eqn:tauJordanproduct} in this basis gives
\be
\tau=\frac{1}{2}\sum_{i,j}(p_{i}+p_{j})|i\>\<j|\otimes\mathcal{E}\big(|j\>\<i|\big).
\ee
Multiplying by $|k\>\<l|\otimes1_{\Alg{B}}$ (on either the right or the left) and taking the partial trace over $\Alg{A}$ gives 
\be
\Tr_{\Alg{A}}\Big[\tau\big(|k\>\<l|\otimes1_{\Alg{B}}\big)\Big]=\frac{p_{k}+p_{l}}{2}\mathcal{E}\big(|k\>\<l|\big), 
\ee
from which the claimed formula~\eqref{eq:qconditionalmap} for $\mathcal{E}$ follows. 
\eprf

A manifestly basis-independent formula for $\mathcal{E}$ can be obtained from~\eqref{eqn:JEfromSylvester} via $\mathcal{E}(A)=\Tr_{\Alg{A}}\big[\Jamio[\mathcal{E}](A\otimes1_{\Alg{B}})\big]$, though we will not make use of such an expression here. 

\bprf[Proof of Theorem~\ref{thm:tempcompSB}]
We will prove this theorem in the special case where $\Alg{A}$ and $\Alg{B}$ are full matrix algebras, say $\Alg{A}=\matr_{m}$ and $\Alg{B}=\matr_{n}$. 
Since $\tau$ is a separable state, it can be expressed as in~\eqref{eqn:sepstate}. 
Then the formula for $\mathcal{E}$ from~\eqref{eq:qconditionalmap} becomes
\be
\mathcal{E}\big(|i\>\<j|\big)=\frac{2}{p_{i}+p_{j}}\sum_{\theta} t_{\theta}\<j|\rho_{\Alg{A};\theta}|i\>\rho_{\Alg{B};\theta}. \label{eqn:temporalchannel}
\ee
The Choi matrix of $\mathcal{E}$ is then given by 
\begin{align}
\Choi[\mathcal{E}]
&=\sum_{i,j}|i\>\<j|\otimes\mathcal{E}\big(|i\>\<j|\big) \nonumber \\
&=\sum_{\theta} t_{\theta}\left(2\sum_{i,j}\frac{\<j|\rho_{\Alg{A};\theta}|i\>}{p_{i}+p_{j}}|i\>\<j|\right)\otimes\rho_{\Alg{B};\theta}.
\label{eqn:ChoiEseparable}
\end{align}
Note that the matrix inside the parentheses can be expressed as the Hadamard--Schur product 
\be
2\sum_{i,j}\frac{\<j|\rho_{\Alg{A};\theta}|i\>}{p_{i}+p_{j}}|i\>\<j|
=2\sum_{i,j}\frac{\<i|\rho_{\Alg{A};\theta}^{T}|j\>}{p_{i}+p_{j}}|i\>\<j|
=\omega\odot\rho_{\Alg{A};\theta}^{T},
\ee
where ${}^{T}$ denotes the transpose with respect to the basis $\{|i\>\}$ and 
\be
\label{eqn:omega}
\omega:=2\sum_{i,j}\frac{|i\>\<j|}{p_{i}+p_{j}}.
\ee 
By Lemma~\ref{lem:invsumprob}, $\omega$ is positive. Since $\rho_{\Alg{A};\theta}$ is positive, $\rho_{\Alg{A};\theta}^{T}$ is positive as well. These two facts combined with the Hadamard--Schur product theorem (Lemma~\ref{lem:Schurproductthm}) implies that $\omega\odot\rho_{\Alg{A};\theta}^{T}$ is positive. Therefore, 
\be
\label{eq:choichannel}
\Choi[\mathcal{E}]=\sum_{\theta} t_{\theta}\left(\omega\odot\rho_{\Alg{A};\theta}^{T}\right)\otimes\rho_{\Alg{B};\theta}
\ee
is a convex combination of positive elements and is therefore positive. Hence, $\mathcal{E}$ is CPTP by Choi's theorem~\cite{Ch75}, thus concluding the proof. 
\eprf

An immediate consequence of Theorem~\ref{thm:tempcompSB} is temporal compatibility in both temporal directions for separable states. 

\bc
\label{cor:tempcompSB}
Let $\Alg{A}=\matr_{m}$ and $\Alg{B}=\matr_{n}$.
Let $\tau$ be a separable density matrix in $\Alg{A}\otimes\Alg{B}$ and set $\rho_{\Alg{A}}=\Tr_{\Alg{B}}[\tau]$ and $\rho_{\Alg{B}}=\Tr_{\Alg{A}}[\tau]$. Assume $\rho_{\Alg{A}}>0$ and $\rho_{\Alg{B}}>0$. Then $\tau$ is temporally compatible in both directions, i.e., there exist unique CPTP maps $\Alg{A}\xrightarrow{\mathcal{E}}\Alg{B}$ and $\Alg{B}\xrightarrow{\mathcal{F}}\Alg{A}$ such that $\tau=\mathcal{E}\star\rho_{\Alg{A}}$ and $\gamma(\tau)=\mathcal{F}\star\rho_{\Alg{B}}$. 
\ec

We can rephrase this corollary in terms of its contrapositive providing a sufficient condition for entanglement. 

\bc
\label{cor:temporalPPTcriterion}
Let $\Alg{A}=\matr_{m}$ and $\Alg{B}=\matr_{n}$.
Let $\tau$ be a density matrix in $\Alg{A}\otimes\Alg{B}$, and suppose $\rho_{\Alg{A}}:=\Tr_{\Alg{B}}[\tau]>0$ and $\rho_{\Alg{B}}:=\Tr_{\Alg{B}}[\tau]>0$. If $\tau$ is \emph{not} temporally compatible in at least one direction, then $\tau$ must be entangled. 
\ec

We note, however, that it is known that there exist entangled states that are temporally compatible~\cite{SNREG23}. Therefore, Corollary~\ref{cor:tempcompSB} only gives a sufficient condition for temporal compatibility. 
We will come back to necessary and sufficient conditions later in this work.

\subsection{Unwrapping the temporal channel associated with a separable state}

In this section, assuming $\Alg{A}=\matr_{m}$ and $\Alg{B}=\matr_{n}$, we will provide an interpretation for the temporal channels $\Alg{A}\xrightarrow{\mathcal{E}}\Alg{B}$ and $\Alg{B}\xrightarrow{\mathcal{F}}\Alg{A}$ coming from a bipartite density matrix $\tau\in\Alg{A}\otimes\Alg{B}$ as in Theorem~\ref{thm:tempcompSB}. Namely, we will show that each of these channels can be expressed as a dephasing channel followed by a measure-and-prepare protocol where the measurement operators form a pretty good measurement~\cite{hausladen1994pretty}. For concreteness, let us demonstrate this with an arbitrary separable state $\tau=\sum_{\theta\in\Theta}t_{\theta}\rho_{\Alg{A};\theta}\otimes\rho_{\Alg{B};\theta}$ and its corresponding channel $\mathcal{E}$ in the temporal direction $\Alg{A}\rightarrow\Alg{B}$ so that $\tau={\mathcal{E}\star\rho_{\Alg{A}}}$. A completely analogous set of results holds for the temporal channel $\Alg{B}\xrightarrow{\mathcal{F}}\Alg{A}$ in the other direction. The main theorem in this section reads as follows. 

\bt
\label{thm:splittemporalmap}
Let $\Alg{A}\xrightarrow{\mathcal{E}}\Alg{B}$  be the temporal channel~\eqref{eq:qconditionalmap} associated with the joint separable density matrix $\tau\in\Alg{A}\otimes\Alg{B}$. Then $\mathcal{E}$ can be expressed as a composite of two channels 
$
    \mathcal{E} = \mathcal{G}\circ\mathcal{D}, 
$
where $\Alg{A}\xrightarrow{\mathcal{D}}\Alg{A}$ represents a generalized dephasing channel and $\Alg{A}\xrightarrow{\mathcal{G}}\Alg{B}$ represents a pretty good measure-and-prepare channel. 
Explicitly, $\mathcal{D}$ is defined on $A\in\Alg{A}$ as 
\be
\label{eq:dephasing}
\mathcal{D}(A) := \sum_{i,j}\frac{2\sqrt{p_ip_j}}{p_{i}+p_{j}}  |i\>\<i| A |j\>\<j|,
\ee
where $\{|i\>\}$ is an orthonormal basis of eigenvectors for $\rho_\Alg{A}$, i.e., $\rho_\Alg{A}=\sum_{i}p_{i}|i\>\<i|$, so that $\{p_{i}\}$ forms a probability distribution.
Secondly, $\mathcal{G}$ is defined on $A\in\Alg{A}$ as
\be
\label{eq:GmapX}
    \mathcal{G}(A) := \Tr_{\Alg{A}}\Big[\tau\Big(\big(\rho_{\Alg{A}}^{-\frac{1}{2}}A\rho_{\Alg{A}}^{-\frac{1}{2}}\big)\otimes1_{\Alg{B}}\Big)\Big].
\ee
Equivalently, 
\be
\label{eq:Gmap}
    \mathcal{G}(A) =\sum_\theta \Tr[G_\theta A] \rho_{\Alg{B};\theta},
\ee
where $\tau=\sum_{\theta\in\Theta} t_{\theta}\rho_{\Alg{A};\theta}\otimes\rho_{\Alg{B};\theta}$ is a convex decomposition of $\tau$ into product states, $G_{\theta}:= t_\theta \rho_{\Alg{A}}^{-\frac{1}{2}} \rho_{\Alg{A};\theta} \rho_\Alg{A}^{-\frac{1}{2}}$, and $\{G_{\theta}\}$ is the pretty good measurement that discriminates the ensemble $\{t_\theta,\rho_{\Alg{A};\theta}\}$ with average state $\sum_{\theta} t_\theta \rho_{\Alg{A};\theta}$ equal to $\rho_\Alg{A}$. 
\et

\begin{figure}
\begin{tikzpicture}[scale=0.98]
\node at (-3,-1) {\includegraphics[width=2.5cm,trim={2.8cm 3.4cm 2.1cm 2.5cm},clip]{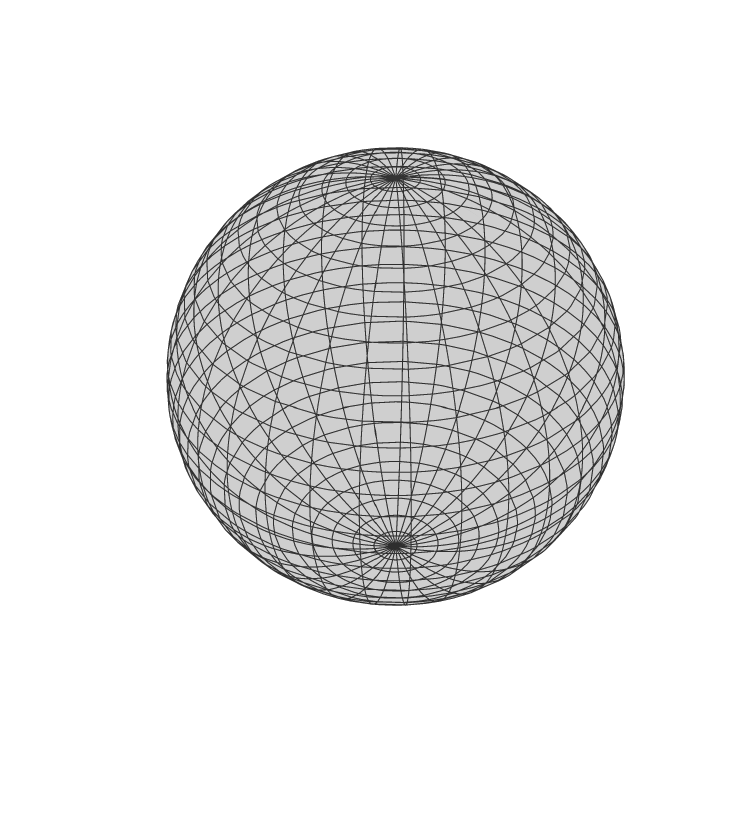} }; 
\node at (0,1) {\includegraphics[width=2.5cm,trim={1.3cm 2.2cm 1.8cm 1.3cm},clip]{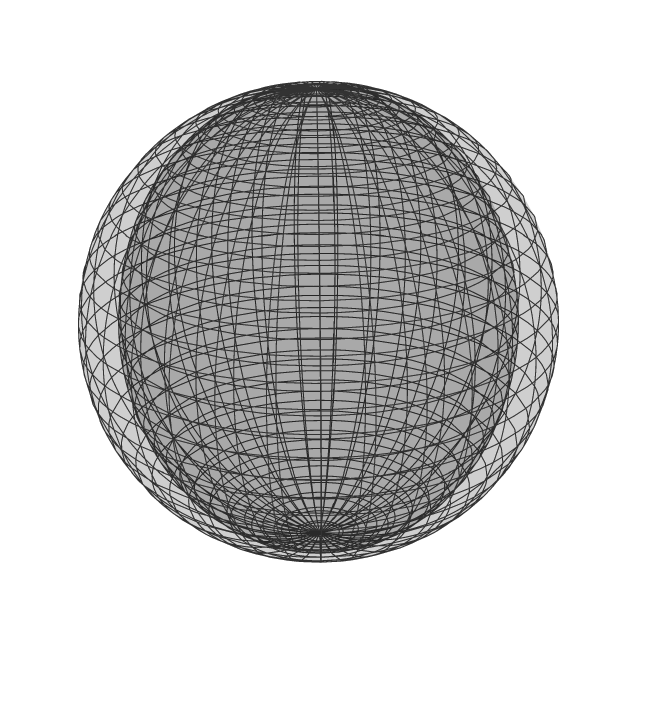} }; 
\node at (3,-1) {\includegraphics[width=2.5cm,trim={2.8cm 3.4cm 2.1cm 2.5cm},clip]{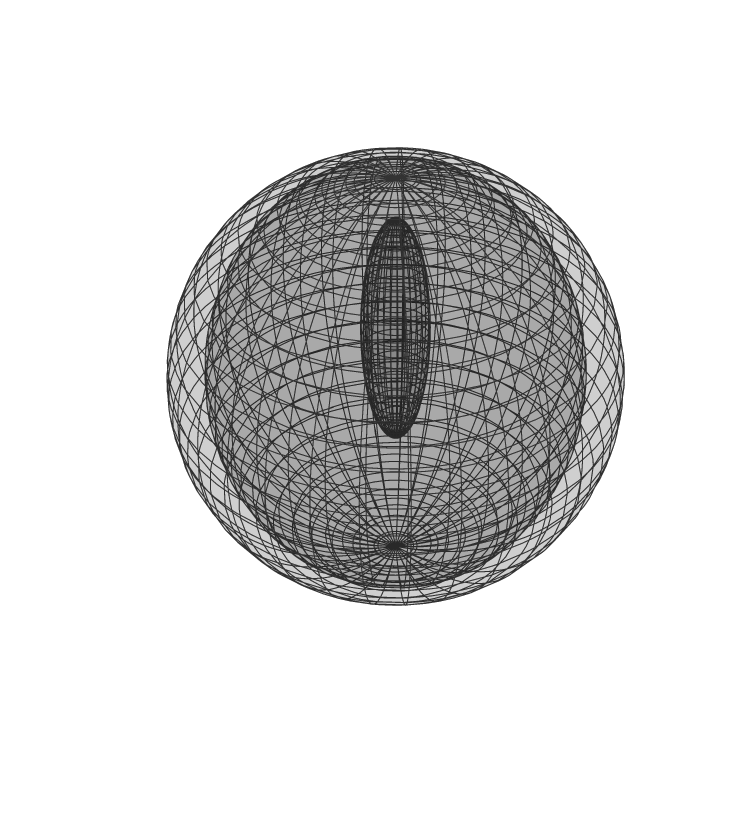} }; 
\draw[very thick,->] (-1.75,-0.25) -- node[above,xshift=-5pt]{$\mathcal{D}$} (-1.25,0.25);
\draw[very thick,->] (1.25,0.25) -- node[above,xshift=5pt]{$\mathcal{G}$} (1.75,-0.25);
\draw[very thick,->] (-1.5,-1) -- node[below]{$\mathcal{E}$} (1.5,-1);
\node at (-2.25,1.25) {\rotatebox{45}{dephasing}};
\node at (2.25,1.25) {\rotatebox{-45}{\begin{tabular}{c}measure\\-and-\\prepare\end{tabular}}};
\node at (0,-1.75) {temporal channel};
\end{tikzpicture}
\caption{A visualization of the decomposition $\mathcal{E}=\mathcal{G}\circ\mathcal{D}$ of the temporal channel in terms of a dephasing channel $\mathcal{D}$ followed by a measure-and-prepare channel $\mathcal{G}$ from Theorem~\ref{thm:splittemporalmap}. To produce these exact images, set 
$\tau=\sum_{\theta}t_{\theta}\rho_{\Alg{A};\theta}\otimes\rho_{\Alg{B};\theta}$, with $\theta\in\{1,\dots,6\}$, $t_{1}=\frac{5}{8}$, $t_{2}=\cdots=t_{6}=\frac{3}{40}$, $\rho_{\Alg{A};1}=\rho_{\Alg{B};1}=|0\>\<0|$, $\rho_{\Alg{A};2}=\rho_{\Alg{B};2}=|1\>\<1|$, $\rho_{\Alg{A};3}=\rho_{\Alg{B};3}=|+\>\<+|$, $\rho_{\Alg{A};4}=\rho_{\Alg{B};4}=|-\>\<-|$, $\rho_{\Alg{A};5}=\rho_{\Alg{B};5}=|\!+i\>\<+i|$, and $\rho_{\Alg{A};6}=\rho_{\Alg{B};6}=|\!-i\>\<-i|$. Our convention here is that $|0\>=(1,0)$, $|1\>=(0,1)$, $|\pm\>=\frac{1}{\sqrt{2}}\big(|0\>\pm|1\>\big)$, and $|\pm i\>=\frac{1}{\sqrt{2}}\big(|0\>\pm i|1\>\big)$. 
The figure on the left is that of the Bloch ball $\Alg{S}$. The figure on top shows the image, $\mathcal{D}(\Alg{S})$, of $\Alg{S}$ under the map $\mathcal{D}$ as well as $\Alg{S}$ for comparison. The figure on the right shows the image of $\Alg{S}$ under $\mathcal{E}$ as well as $\Alg{S}$ and $\mathcal{D}(\Alg{S})$ for comparison (the figure on the right does not show $\mathcal{G}(\Alg{S})$ to avoid cluttering). The measure-and-prepare channel illustrates greater distinguishability in the vertical direction due to the higher weight given by $t_{1}+t_{2}$ as compared with $t_{3}+t_{4}=t_{5}+t_{6}$.
}
\label{fig:EequalsGD}
\end{figure}
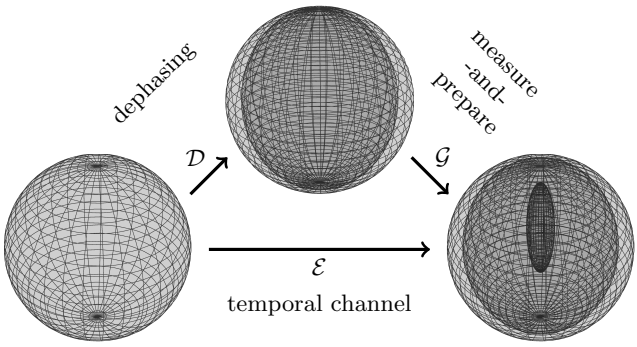

An immediate consequence of Theorem~\ref{thm:splittemporalmap} is that $\mathcal{E}$ is entanglement breaking whenever $\tau$ is separable~\cite{horodecki2003ebc}. A visualization of the channels $\mathcal{D}$ and $\mathcal{G}$ of Theorem~\ref{thm:splittemporalmap} for an illustrative example is shown in Figure~\ref{fig:EequalsGD}. We note that just like $\mathcal{E}$ in~\eqref{eq:qconditionalmap} is well-defined in the sense that it does not depend on the choice of an orthonormal basis of eigenvectors for $\rho_{\Alg{A}}$, for example as shown by its alternative expression in~\eqref{eqn:JEfromSylvester}, the map $\mathcal{D}$ in~\eqref{eq:dephasing} is also well-defined, depending only on $\rho_{\Alg{A}}$, but not the particular choice of orthonormal eigenbasis (cf.\ Appendix~\ref{app:Dwelldefined}).    
Moreover, note that~\eqref{eq:GmapX} ensures that $\mathcal{G}$ does not depend on the choice of a convex decomposition of $\tau$. 
Generalizations of Theorem~\ref{thm:splittemporalmap} for when $\rho_{\Alg{A}}\ge0$ are given in Appendix~\ref{sec:nonfaithful}.

The remainder of this section is devoted to making Theorem~\ref{thm:splittemporalmap} more understandable and then proving it. In particular, we will analyze the different components and justify our terminology before providing a proof. To do so, we first recall some definitions involving Hadamard--Schur channels and correlation matrices. 

A \define{correlation matrix} is a hermitian positive semidefinite matrix whose diagonal entries are all equal to $1$~\cite{LiTam1994}. Note that it follows from this definition that the magnitudes of the off-diagonal entries are between $0$ and $1$ (otherwise the subdeterminant associated with that two-dimensional subspace would become negative violating the subdeterminant test for positive semidefinite matrices~\cite{Strang2022}). 
Equivalently, an $n\times n$ matrix $C$ is a correlation matrix whenever there exist $n$ normalized vectors $|v_{i}\>\in\C^{n}$ such that $\<i|C|j\>=\<v_{i}|v_{j}\>$ for all $i,j=1,\dots,n$, i.e., 
\be
\label{eqn:Grammatrix}
C=\sum_{i,j}\<v_{i}|v_{j}\>|i\>\<j|=\begin{bmatrix}\text{\raisebox{5pt}{\uline{\hspace{4mm}}}}\!&\!v_{1}^{\dag}\!&\!\text{\raisebox{5pt}{\uline{\hspace{4mm}}}}\\&\vdots&\\\text{\raisebox{5pt}{\uline{\hspace{4mm}}}}\!&\!v_{n}^{\dag\!}&\!\text{\raisebox{5pt}{\uline{\hspace{4mm}}}}\end{bmatrix}
\begin{bmatrix}|&&|\\v_{1}&\cdots&v_{n}\\|&&|\end{bmatrix}.
\ee
This equivalence follows directly from the Cholesky test for positive semidefinite matrices~\cite{Strang2022}. 
In other words, an $n\times n$ correlation matrix is the \emph{Gram matrix} associated with a set of $n$ normalized vectors. 
In what follows, we will say that a correlation matrix is \define{strict} when it is strictly positive, which implies its off-diagonal entries have magnitude strictly less than $1$ (otherwise the subdeterminant would vanish). Note that a correlation matrix $C$ is strict whenever there exists a normalized \emph{basis} (though not necessarily orthonormal) $\{|v_{i}\>\}$ such that $\<i|C|j\>=\<v_{i}|v_{j}\>$ for all $i,j$. 
Using a correlation matrix with the Hadamard--Schur product yields a \emph{Hadamard--Schur channel}, which is a quantum channel that is given by the Hadamard--Schur product with some matrix with respect to some orthonormal basis~\cite{Watrous18}. The following is the special case when the matrix is a correlation matrix. 

\begin{definition}
Let $\Alg{A}=\matr_{m}$. A \define{generalized dephasing channel} on $\Alg{A}$ is a Hadamard--Schur channel $\Alg{A}\xrightarrow{\mathcal{D}}\Alg{A}$ of the form 
\be
\label{eqn:generalizeddephasing}
\mathcal{D}(A)=C\odot A=
\sum_{i,j}\<i|C|j\>\<i|A|j\> |i\>\<j|,
\ee
where $C$ is some correlation matrix with respect to some orthonormal basis, called the \define{dephasing basis}, $\{|i\>\}$ of $\C^{m}$. 
\end{definition}

We note that a dephasing channel $\mathcal{D}$ as defined in~\eqref{eqn:generalizeddephasing} is in fact completely positive by Proposition 4.17 in Ref.~\cite{Watrous18}. Moreover, $\mathcal{D}$ is also trace-preserving since 
\begin{align}
\Tr\big[\mathcal{D}(A)\big]&=
\sum_{i,j,k}\<i|C|j\>\<i|A|j\>\<k|i\>\<j|k\> \nonumber \\
&=
\sum_{i}\<i|C|i\>\<i|A|i\>=\Tr[A]
\end{align}
for all $A\in\Alg{A}$ because $\<i|C|i\>=1$ for all $i$. Hence, $\mathcal{D}$ is a quantum channel. 

Our definition of a generalized dephasing channel is equivalent to that given in Refs.~\cite{ToWiWi17,DeSh05} and Section~5.2 of Ref.~\cite{WildeQIT16} due to the equivalence between a correlation matrix and its expression as a Gram matrix as in~\eqref{eqn:Grammatrix}. The motivation for this definition of a generalized dephasing channel combines the following observations. First, for any density matrix $\rho_{d}=\sum_{i}p_{i}|i\>\<i|$ that is diagonal in the dephasing basis, we have $\mathcal{D}(\rho_{d})=\rho_{d}$. 
Moreover, if $C$ is a correlation matrix, then for an arbitrary density matrix $\rho$ (possibly with off-diagonal terms), 
$
\mathcal{D}(\rho)=
\sum_{i,j}\<i|C|j\>\<i|\rho|j\>|i\>\<j|.
$
When $C$ is a strict correlation matrix, the off-diagonal terms of $\rho$ decrease in magnitude after the application of $\mathcal{D}$ since $|\<i|C|j\>|<1$ for all $i\ne j$. Hence, by the strictness condition on the correlation matrix, 
\be
\lim_{n\to\infty}(\underbrace{\mathcal{D}\circ\cdots\circ\mathcal{D}}_{\text{$n$ times}})(\rho)=\sum_{i=1}^{m}\<i|\rho|i\>|i\>\<i|
\ee
for all density matrices $\rho$ in $\Alg{A}$. Thus, iterating a generalized dephasing channel causes any input density matrix to converge towards the density matrix obtained from just its diagonal entries, thereby removing coherence terms in $\rho$ with respect to the dephasing basis. 
If $C$ is not strict, then certain subspaces formed by the dephasing basis cause density matrices to retain coherences within those subspaces.

\begin{lemma}
\label{lem:harmonicmean}
Given any nowhere-vanishing probability distribution $\{p_{i}\}$ and an orthonormal basis $\{ |i\>\}$, the matrix 
\be
\label{eq:harmonicmeanmatrix}
H=\sum_{i,j}\frac{2\sqrt{p_i p_j}}{p_i+p_j}|i\>\<j|
\ee
is a correlation matrix. 
\end{lemma}

The $ij$ entry $\<i|H|j\>$ of the matrix $H$ in Lemma~\ref{lem:harmonicmean} represents the \emph{harmonic mean} of $p_i$ and $p_j$~\cite{Bh07}.

\begin{proof}[Proof of Lemma~\ref{lem:harmonicmean}]
First note that the diagonal entries of $H$ are all equal to $1$. The rest follows from the arithmetic-geometric mean inequality
\be
\frac{x+y}{2}\ge\sqrt{xy}, 
\ee
which holds for all nonnegative real numbers $x,y$~\cite{Bh07}. 
\end{proof}

Having established these definitions and lemmas, notice that $\mathcal{D}$ from Equation~\eqref{eq:dephasing} is a generalized dephasing channel of the form
\be
\label{eq:gendephaseasSchur}
\mathcal{D}(A) = H \odot A, 
\ee
where $H$ is given by~\eqref{eq:harmonicmeanmatrix} and $A\in\Alg{A}$ is arbitrary.  

The next ingredient in Theorem~\ref{thm:splittemporalmap} is the map $\mathcal{G}$, which represents a pretty good measure-and-prepare channel and appears in the context of quantum state discrimination protocols~\cite{BaeKwek15}, which we discuss later in Section~\ref{sec:statediscrim}. We first recall measure-and-prepare channels. 

Given algebras $\Alg{A}$ and $\Alg{B}$, a \define{measure-and-prepare channel} from $\Alg{A}$ to $\Alg{B}$ is a channel $\Alg{A}\xrightarrow{\mathcal{E}}\Alg{B}$ of the form 
\be
\label{eq:temporalchannelMPform}
    \mathcal{E}(A) = \sum_\theta \Tr[E_{\theta} A] \rho_{\Alg{B};\theta},
\ee
where $\theta$ is an element of a finite index set, $\{E_{\theta}\}$ is a positive operator-valued measure (POVM)~\cite{Kr83,NiCh11}, and $\{\rho_{\Alg{B};\theta}\}$ is a collection of density matrices in $\Alg{B}$. 
Such a channel can be interpreted as taking an input state $\rho$, measuring outcome $\theta$ with probability $\Tr[E_{\theta}\rho]$, and then preparing the state $\rho_{\Alg{B};\theta}$ conditioned on measuring the outcome $\theta$. Note that the Choi matrix of a measure-and-prepare channel is of the form
\be
\label{eqn:ChoiMandP}
\Choi[\mathcal{E}]=\sum_{\theta}E_{\theta}^{T}\otimes \rho_{\Alg{B};\theta},
\ee
where the transpose is computed with respect to the same basis that the Choi matrix is defined in terms of.

By comparing~\eqref{eqn:ChoiMandP} to the Choi matrix of the temporal channel $\mathcal{E}$ in~\eqref{eq:choichannel}, we deduce that 
\be
\big\{E_{\theta}:=t_{\theta}(\omega\odot\rho_{\Alg{A};\theta})\big\}
\ee
defines a POVM, thereby giving an interpretation to $\mathcal{E}$ as a measure-and-prepare channel of the form~\eqref{eq:temporalchannelMPform}. 
However, as we will see soon, we can provide a clearer interpretation of $\mathcal{E}$ by splitting it up into a dephasing channel, as above, followed by a well-known measure-and-prepare channel given by the \emph{pretty good measurement}~\cite{hausladen1994pretty}.
We recall this definition next. 

Given a finite index set $\Theta$ and an algebra $\Alg{A}$, a \define{probabilistic ensemble of states} in $\Alg{A}$ is a collection $\{t_\theta,\rho_{\Alg{A};\theta}\}$, where $\{t_{\theta}\}$ defines a probability distribution and each $\rho_{\Alg{A};\theta}$ is a density matrix in $\Alg{A}$. The \define{average state} associated with this ensemble is $\rho_{\Alg{A}}:=\sum_{\theta}t_{\theta}\rho_{\Alg{A};\theta}$. The \define{pretty good measurement} associated with this ensemble is the POVM $\{G_{\theta}\}$, where 
\be
\label{eq:PGMPOVM}
G_{\theta}:= t_{\theta}\rho_{\Alg{A}}^{-\frac{1}{2}}\rho_{\Alg{A};\theta}\rho_{\Alg{A}}^{-\frac{1}{2}}.
\ee
We still assume $\rho_{\Alg{A}}>0$ here for simplicity (see Appendix~\ref{sec:nonfaithful} for the more general case $\rho_{\Alg{A}}\ge0$). 
Because $\{G_{\theta}\}$ defines a POVM, given a collection $\{\rho_{\Alg{B};\theta}\}$ of density matrices in $\Alg{B}$, the associated measure-and-prepare map 
$\mathcal{G}(A)=\sum_{\theta}\Tr[G_{\theta}A]\rho_{\Alg{B};\theta}$
is indeed a quantum channel. 

We now have all the ingredients needed to prove Theorem~\ref{thm:splittemporalmap}. 

\begin{proof}[Proof of Theorem~\ref{thm:splittemporalmap}]
Since we have already proved that $\mathcal{D}$ and $\mathcal{G}$ from the statement of Theorem~\ref{thm:splittemporalmap} are CPTP maps in the preceeding paragraphs, here we prove that $\mathcal{E}$ is the composite of $\mathcal{D}$ followed by $\mathcal{G}$. Using the measure-and-prepare formula~\eqref{eq:temporalchannelMPform} for $\mathcal{E}$, we obtain 
\begingroup
\allowdisplaybreaks
\begin{align}
    &\mathcal{E}(A) = \sum_\theta \Tr\left[ \sum_{i,j}\frac{2\<j| t_\theta \rho_{\Alg{A};\theta}|i\>}{p_{i}+p_{j}}|j\>\<i| A\right] \rho_{\Alg{B};\theta} \nonumber \\
    &= \sum_\theta \Tr\left[ t_\theta \rho_{\Alg{A};\theta} \sum_{i,j}\frac{2}{p_{i}+p_{j}} |i\>\<i| A |j\>\<j|\right] \rho_{\Alg{B};\theta} \nonumber \\
    &= \sum_\theta \Tr\left[ G_{\theta}\sum_{i,j}\frac{2}{p_{i}+p_{j}} \rho^{\frac{1}{2}}_\Alg{A} |i\>\<i| A |j\>\<j|\rho^{\frac{1}{2}}_\Alg{A}\right]  \rho_{\Alg{B};\theta} \nonumber \\
    &= \sum_\theta \Tr\left[ G_{\theta}\sum_{i,j}\frac{2\sqrt{p_ip_j}}{p_{i}+p_{j}}  |i\>\<i| A |j\>\<j|\right]  \rho_{\Alg{B};\theta}\nonumber \\
    &=(\mathcal{G}\circ\mathcal{D})(A), 
\end{align}
\endgroup
where $\{p_{i}\}$ are the eigenvalues of $\rho_{\Alg{A}}$, the second and third equalities hold by cyclicity of trace, and  $G_{\theta}=t_{\theta}\rho_{\Alg{A}}^{-\frac{1}{2}}\rho_{\Alg{A};\theta}\rho_{\Alg{A}}^{-\frac{1}{2}}$ is as in~\eqref{eq:PGMPOVM}. 
\end{proof}

\br
In the context of Theorem~\ref{thm:splittemporalmap}, if $\tau\in\Alg{A}\otimes\Alg{B}$ is a general state (not necessarily separable), it is still the case that $\mathcal{E}=\mathcal{G}\circ\mathcal{D}$, where $\mathcal{D}$ is given by~\eqref{eq:dephasing} and $\mathcal{G}$ is given by~\eqref{eq:GmapX}. However, $\mathcal{G}$ need not be a quantum channel when $\tau$ is an arbitrary state. Nevertheless, $\mathcal{G}$ is still a positive trace-preserving map (so it need not be \emph{completely} positive). Positivity of $\mathcal{G}$ can be seen by rewriting it as 
\be
\mathcal{G}(A)=\Tr_{\Alg{A}}\Big[\Big(\rho_{\Alg{A}}^{-\frac{1}{2}}A^{\frac{1}{2}}\otimes1_{\Alg{B}}\Big)^{\dag}\tau\Big(\rho_{\Alg{A}}^{-\frac{1}{2}}A^{\frac{1}{2}}\otimes1_{\Alg{B}}\Big)\Big],
\ee 
which is positive because $\Tr_{\Alg{A}}$ is a positive map and because the term inside the square brackets is positive for all positive $A\in\Alg{A}$. 
This extends the results of Ref.~\cite{PaQPL21}, since a natural followup question to that work is if positive conditionals defined using~\eqref{eq:Estarrho} (cf.\ Remark~\ref{rmk:conditionals}) always exist for bipartite states. The answer is yes, because  $\mathcal{E}=\mathcal{G}\circ\mathcal{D}$ is a composite of two positive trace-preserving maps. 
\er

Before concluding this subsection, we briefly discuss temporal compatibility in the other direction. 
For the other temporal order with $\Alg{B}$ having a direct influence on $\Alg{A}$, a similar interpretation holds for the temporal channel $\Alg{B}\xrightarrow{\mathcal{F}}\Alg{A}$ constructed in a way analogous to $\mathcal{E}$. Namely, $\mathcal{F}$ decomposes as  
$
    \mathcal{F} = \mathcal{G}'\circ\mathcal{D}',
$
where $\Alg{B}\xrightarrow{\mathcal{D}'}\Alg{B}$ is a dephasing channel with respect to an orthonormal basis of eigenvectors for $\rho_\Alg{B}$ and $\Alg{B}\xrightarrow{\mathcal{G}'}\Alg{A}$ is the pretty good measure-and-prepare channel defined on $B\in\Alg{B}$ by 
\begin{align}
\mathcal{G}'(B)&=\Tr_{\Alg{B}}\Big[\tau\Big(1_{\Alg{A}}\otimes \big(\rho_{\Alg{B}}^{-\frac{1}{2}}B\rho_{\Alg{B}}^{-\frac{1}{2}}\big)\Big)\Big]\nonumber\\
&=\sum_\theta t_\theta \Tr\Big[\rho^{-\frac{1}{2}}_\Alg{B} \rho_{\Alg{B};\theta} \rho^{-\frac{1}{2}}_\Alg{B}B\Big] \rho_{\Alg{A};\theta}.
\label{eqn:temporalGreverse}
\end{align}

\subsection{Temporal maps and retrodiction}

Interestingly, the two channels $\mathcal{G}$ and $\mathcal{G}'$ from~\eqref{eq:Gmap} and~\eqref{eqn:temporalGreverse}, respectively, are Petz recovery maps of each other with the reduced states of $\tau$ being respective prior states~\cite{Connes74,AcCe82,Pe84,OhPe93,PaBu22}. More precisely, the process $(\mathcal{G}',\rho_{\Alg{B}})\in\mathscr{P}(\Alg{B},\Alg{A})$ is the Petz recovery process of $(\mathcal{G},\rho_{\Alg{A}})\in\mathscr{P}(\Alg{A},\Alg{B})$. This can be shown as follows. 

First note that the Hilbert--Schmidt adjoint $\mathcal{G}^*$ of the channel  $\mathcal{G}(A)=\sum_\theta t_\theta \Tr\big[\rho^{-\frac{1}{2}}_\Alg{A} \rho_{\Alg{A};\theta} \rho^{-\frac{1}{2}}_\Alg{A} A\big] \rho_{\Alg{B};\theta}$ is given on $B\in\Alg{B}$ by 
\be
\label{eqn:HSadjointsG}
\mathcal{G}^*(B)=\sum_\theta t_\theta\Tr[ \rho_{\Alg{B};\theta} B] \rho^{-\frac{1}{2}}_\Alg{A} \rho_{\Alg{A};\theta} \rho^{-\frac{1}{2}}_\Alg{A}.
\ee

Now let $\hat{\mathcal{G}}_{\rho_\Alg{A}}$ be the Petz recovery map of $\mathcal{G}$ with respect to the prior state ${\rho_\Alg{A}}$, i.e., 
\be
\hat{\mathcal{G}}_{\rho_{\Alg{A}}}(B):=\rho_\Alg{A}^{\frac{1}{2}}\mathcal{G}^*\Big(\mathcal{G}(\rho_\Alg{A})^{-\frac{1}{2}}B\mathcal{G}(\rho_\Alg{A})^{-\frac{1}{2}}\Big)\rho_\Alg{A}^{\frac{1}{2}} 
\ee
for all $B\in\Alg{B}$. 
Since $\mathcal{G}(\rho_\Alg{A})=\rho_\Alg{B}$, we have
\begin{align}
    &\hat{\mathcal{G}}_{\rho_\Alg{A}}(B)
    =\rho_\Alg{A}^{\frac{1}{2}}\mathcal{G}^*\Big(\rho_\Alg{B}^{-\frac{1}{2}}B\rho_\Alg{B}^{-\frac{1}{2}}\Big)\rho_\Alg{A}^{\frac{1}{2}}\nonumber \\
    &=\rho_\Alg{A}^{\frac{1}{2}}\left(\sum_\theta t_\theta \Tr\Big[ \rho_{\Alg{B};\theta}\rho_\Alg{B}^{-\frac{1}{2}}B\rho_\Alg{B}^{-\frac{1}{2}}\Big]\rho_\Alg{A}^{-\frac{1}{2}} \rho_{\Alg{A};\theta} \rho_\Alg{A}^{-\frac{1}{2}}\right)\rho_\Alg{A}^{\frac{1}{2}} \nonumber\\
    &=\sum_\theta t_\theta \Tr\Big[\rho_\Alg{B}^{-\frac{1}{2}} \rho_{\Alg{B};\theta}\rho_\Alg{B}^{-\frac{1}{2}}B\Big] \rho_{\Alg{A};\theta} = \mathcal{G}'(B), \label{eq:GhatGprime}
\end{align}
where we have used~\eqref{eqn:HSadjointsG} in the second equality and~\eqref{eqn:temporalGreverse} in the fourth equality. This proves that $\mathcal{G}'$ is the Petz recovery map of $\mathcal{G}$ with respect to the prior state $\rho_{\Alg{A}}$. Similarly, $\mathcal{G}$ is the Petz recovery map of $\mathcal{G}'$ with respect to the prior state $\rho_\Alg{B}$.

Despite this relationship between $\mathcal{G}$ and $\mathcal{G}'$, we note that $\mathcal{E}$ and $\mathcal{F}$ are \emph{not} in general Petz recovery maps of each other. Rather, they are \emph{Bayesian inverses} of each other based on the notion of Bayesian inverse associated with the canonical state over time~\cite{FuPa22a}, which is reviewed in Appendix~\ref{app:generaltemporalcompatibility}. Moreover, a relationship between the two is given by
\be
\label{eq:DFED}
\mathcal{D}\circ\mathcal{F}=\mathcal{D}\circ\mathcal{G}'\circ\mathcal{D}'
=\hat{\mathcal{D}}_{\rho_{\Alg{A}}}\circ\hat{\mathcal{G}}_{\rho_{\Alg{A}}}\circ\mathcal{D}'
=\hat{\mathcal{E}}_{\rho_{\Alg{A}}}\circ\mathcal{D}',
\ee
where the first equality follows from Theorem~\ref{thm:splittemporalmap} for $\mathcal{F}$, the second equality follows from~\eqref{eq:GhatGprime} and the fact that the Petz recovery map of the generalized dephasing channel is itself,  and the third equality follows from Theorem~\ref{thm:splittemporalmap} for $\mathcal{E}$ and functoriality of the Petz recovery map~\cite{PaBu22,Pa24,Ts22b,Wilde15}. Equation~\eqref{eq:DFED} illustrates how the Petz recovery map $\hat{\mathcal{E}}_{\rho_{\Alg{A}}}$ and Bayesian inverse $\mathcal{F}$ of $(\mathcal{E},\rho_{\Alg{A}})$ are related through the action of generalized dephasing channels. 
Therefore, although the Petz recovery map is often considered as the canonical expression extending Bayesian inversion and retrodiction to the quantum setting~\cite{AwBuSc21,BuSc21,KK19,PaBu22,LiWi18,CHPSSW19,LeSp13,AZBS24,CoSp12}, more general Bayesian inverses as defined in Ref.~\cite{FuPa22a} (see also Refs.~\cite{Ts22,Ts22b,SASDS23,PaFu24TSC}) still admit several key characteristics of retrodiction~\cite{BPJ00,ABL64,Wat55,JeOiBr24}.

\subsection{Relation to state discrimination}
\label{sec:statediscrim}

We next contrast the implications of Theorems~\ref{thm:tempcompSB} and~\ref{thm:splittemporalmap} with quantum state discrimination due to the somewhat surprising appearance of the pretty good measurement~\cite{BaeKwek15,hausladen1994pretty}. In brief, our results say that for a separable state, the observed measurement statistics of light-touch observables are compatible with a temporal structure. Our results do not say that quantum states can be prepared in a temporal way, which would otherwise conflict with the results of quantum state discrimination. This section will make these statements more precise. 

A typical way to prepare a separable state in a spatially compatible way is by using a classical random number generator and announcing the result to both Alice and Bob, who may be spatially separated from each other. Upon receiving the outcome of this random classical output, they individually prepare respective quantum states based on the announced random number. Mathematically, this is modeled by a finite set $\Theta$ with probability distribution $t_{\theta}$, which represents the random number generator, together with collections of states $\rho_{\Alg{A};\theta}$ and $\rho_{\Alg{B};\theta}$, which are the states that Alice and Bob prepare, respectively, given the outcome $\theta\in\Theta$ of the random number generator. Diagrammatically, this spatial structure is represented as 
\[
\vcenter{\hbox{%
\begin{tikzpicture}[font=\small]
\node[state] (omega) at (0,0) {\;\;$\tau$\;\;};
\coordinate (A) at (-0.5,1.05);
\coordinate (B) at (0.5,1.05);
\draw (omega) ++(-0.5,0) -- (A);
\draw (omega) ++(0.5,0) -- (B);
\path[scriptstyle]
node at (-0.35,0.85) {$\Alg{A}$}
node at (0.65,0.85) {$\Alg{B}$};
\end{tikzpicture}
}}
\quad=\quad
\vcenter{\hbox{%
\begin{tikzpicture}[font=\small]
\node[state] (p) at (0,-0.2) {$t$};
\node[copier] (c) at (0,0.3) {};
\node[arrow box] (P) at (-0.5,0.95) {$\mathcal{P}$};
\node[arrow box] (Q) at (0.5,0.95) {$\mathcal{Q}$};
\coordinate (A) at (-0.5,1.65);
\coordinate (B) at (0.5,1.65);
\draw[double] (p) to (c);
\draw[double] (c) to[out=165,in=-90] (P);
\draw[double] (c) to[out=15,in=-90] (Q);
\draw (P) -- (A);
\draw (Q) -- (B);
\path[scriptstyle]
node at (-0.35,1.50) {$\Alg{A}$}
node at (0.65,1.50) {$\Alg{B}$}
node at (0.25,0.10) {$\C^{\Theta}$};
\end{tikzpicture}
}}
\]
and visually captures the spatial aspect of the causal structure associated with the expression of a separable state $\tau\in\Alg{A}\otimes\Alg{B}$ as $\tau=\sum_{\theta\in\Theta} t_{\theta}\rho_{\Alg{A};\theta}\otimes\rho_{\Alg{B};\theta}$. In this picture, $\C^{\Theta}\xrightarrow{\mathcal{P}}\Alg{A}$ and $\C^{\Theta}\xrightarrow{\mathcal{Q}}\Alg{B}$ are the CPTP maps associated with the preparations of $\{\rho_{\Alg{A};\theta}\}$ and $\{\rho_{\Alg{B};\theta}\}$, respectively~\cite{FuPa22a}. Also, the trivalent vertex represents the copy map~\cite{Fr20}, which means that the same outcome $\theta$ is sent to both Bob and Alice.

At a glance, Theorem~\ref{thm:tempcompSB} seems to imply that the separable state $\tau$ can be prepared in a temporal way. Namely, such a temporal preparation protocol would involve Alice sending a copy of her state to Bob, who then acts on that state by the temporal channel $\mathcal{E}$ from Theorem~\ref{thm:tempcompSB} to produce the bipartite state $\tau$. More precisely, Alice first prepares two copies of the quantum state $\rho_{\Alg{A};\theta}$ with probability $t_{\theta}$. Instead of announcing the generated number $\theta$ to Bob, she sends one copy of the quantum state $\rho_{\Alg{A};\theta}$ to Bob. Then, Bob applies the measure-and-prepare channel $\mathcal{E}$ to his received (but unknown) state $\rho_{\Alg{A};\theta}$ as an attempt to carry out a state discrimination protocol in order to identify $\theta$ so that he may successfully prepare the corresponding quantum state $\rho_{\Alg{B};\theta}$. Such a structure would be depicted diagrammatically as 
\[
\vcenter{\hbox{%
\begin{tikzpicture}[font=\small]
\node[state] (p) at (0,-0.2) {$t$};
\node[copier] (c) at (0,0.3) {};
\node[arrow box] (P) at (-0.5,0.95) {$\mathcal{P}$};
\node[arrow box] (P2) at (0.5,0.95) {$\mathcal{P}$};
\node[arrow box] (E) at (0.5,2.00) {$\mathcal{E}$};
\coordinate (A) at (-0.5,2.70);
\coordinate (A2) at (0.5,1.65);
\coordinate (B) at (0.5,2.70);
\draw[double] (p) to (c);
\draw[double] (c) to[out=165,in=-90] (P);
\draw[double] (c) to[out=15,in=-90] (P2);
\draw (P) -- (A);
\draw (P2) -- (A2);
\draw (A2) -- (E);
\draw (E) -- (B);
\path[scriptstyle]
node at (-0.35,2.55) {$\Alg{A}$}
node at (0.65,1.50) {$\Alg{A}$}
node at (0.65,2.55) {$\Alg{B}$}
node at (0.25,0.10) {$\C^{\Theta}$};
\end{tikzpicture}
}}
\quad\text{ with }\quad
\vcenter{\hbox{%
\begin{tikzpicture}[font=\small]
\node[arrow box] (E) at (0,0) {$\mathcal{E}$};
\draw (0,-0.75) -- (E);
\draw (E) -- (0,0.75);
\path[scriptstyle]
node at (0.15,-0.65) {$\Alg{A}$}
node at (0.15,0.65) {$\Alg{B}$};
\end{tikzpicture}
}}
=
\vcenter{\hbox{%
\begin{tikzpicture}[font=\small]
\node[arrow box] (M) at (0,-0.55) {};
\node[arrow box] (P) at (0,0.55) {};
\draw (0,-1.30) -- (M);
\draw[double] (M) -- (P);
\draw (P) -- (0,1.30);
\path[scriptstyle]
node at (0.15,-1.15) {$\Alg{A}$}
node at (0.15,1.15) {$\Alg{B}$};
\end{tikzpicture}
}}
\]
because a measure and prepare channel can be viewed as the composite of a quantum-to-classical channel followed by a classical-to-quantum channel~\cite{FuPa22a}. 
Having this equal to the bipartite state $\tau$ would go against the fact that nonorthogonal quantum states cannot be perfectly distinguished (cf. Section 2.2.4 in Ref~\cite{NiCh11}). And since the states $\{\rho_{\Alg{A};\theta}\}$ need not be orthogonal, such a temporal description of $\tau$ is not possible. In fact, there does not exist \emph{any} measure-and-prepare channel $\Alg{A}\xrightarrow{\mathcal{E}'}\Alg{B}$ such that $\mathcal{E}'(\rho_{\Alg{A};\theta})=\rho_{\Alg{B};\theta}$ for all $\theta$ whenever the $\{\rho_{\Alg{A};\theta}\}$ are nonorthogonal states, where we use the notation $\mathcal{E}'$ to avoid conflict with the channel $\mathcal{E}$ from~\eqref{eq:temporalchannelMPform}. For completeness, we include the precise statement and proof in Appendix~\ref{app:statediscrim}.

The meaning, therefore, of the temporal channel $\mathcal{E}$ in Theorem~\ref{thm:tempcompSB} is that it induces the same \emph{measurement statistics} of light-touch observables obtained in a temporal manner by sequential measurement as would be obtained in a spatial manner by joint measurements on the bipartite state. Therefore, Theorem~\ref{thm:tempcompSB} does not conflict with the known results in quantum state discrimination.

\subsection{Entanglement does not imply temporal incompatibility}

In this section, we generalize Theorem~\ref{thm:tempcompSB} by working with \emph{arbitrary} bipartite states $\tau\in\Alg{A}\otimes\Alg{B}$, not necessarily separable states. More precisely, we provide necessary and sufficient conditions for a general state $\tau$ to be temporally compatible. 
In doing so, we will still utilize the map $\mathcal{D}$ from Theorem~\ref{thm:splittemporalmap}. Note that $\mathcal{D}$ only depends on $\rho_{\Alg{A}}=\Tr_{\Alg{B}}[\tau]$ and is therefore a quantum channel whether $\tau$ is separable or entangled. 

\bt
\label{thm:TempCompat}
Let $\tau$ be a density matrix in $\Alg{A}\otimes\Alg{B}$ and set $\rho_{\Alg{A}}=\Tr_{\Alg{B}}[\tau]$. Assume $\rho_{\Alg{A}}$ is a faithful state, i.e., $\rho_{\Alg{A}}>0$. Then there exists a unique CPTP map $\Alg{A}\xrightarrow{\mathcal{E}}\Alg{B}$ such that $\tau=\mathcal{E}\star\rho_{\Alg{A}}$ if and only if the partial transpose of $(\mathcal{D}\otimes\id_{\Alg{B}})(\tau)$ is positive, i.e., 
\be
\label{eqn:decoheredPPT}
\big((\mathcal{T}\circ\mathcal{D})\otimes \id_{\Alg{B}}\big)(\tau) \ge 0, 
\ee
where $\mathcal{T}:\Alg{A}\to\Alg{A}$ denotes the transpose map (with respect to any orthonormal basis) 
and $\mathcal{D}$ denotes the generalized dephasing channel associated with $\rho_{\Alg{A}}$ from Theorem~\ref{thm:splittemporalmap}, specifically~\eqref{eq:dephasing} or~\eqref{eq:gendephaseasSchur}. 
\et

Before proving this theorem, we make some remarks. First, if the partial transpose is defined with respect to the orthonormal basis used for decomposing $\rho_{\Alg{A}}$, then condition~\eqref{eqn:decoheredPPT} can be rewritten as 
\begin{equation}
\label{eqn:temporalPPT}
(\mathcal{D}\otimes\id_{\Alg{B}})(\tau^{T_{\Alg{A}}})\ge0,
\end{equation}
where $\tau^{T_{\Alg{A}}}$ is the partial transpose of $\tau$ on the $\Alg{A}$ factor. This follows from the fact that $\mathcal{T}\circ \mathcal{D}=\mathcal{D}\circ \mathcal{T}$ when $\mathcal{T}$ is defined use the same basis as $\mathcal{D}$. Note also that if $\mathcal{T}':\Alg{B}\to\Alg{B}$ is the transpose map on $\Alg{B}$ (with respect to any basis), then condition~\eqref{eqn:decoheredPPT} is equivalent to $(\mathcal{D}\otimes\mathcal{T}')(\tau)\ge0$. This follows by applying $\mathcal{T}\otimes\mathcal{T}'$ to the expression~\eqref{eqn:decoheredPPT}, since $\mathcal{T}\otimes\mathcal{T}'$ is a positive map. 

Second, a similar statement holds for temporal compatibility in the other direction. Namely, if $\rho_{\Alg{B}}:=\Tr_{\Alg{A}}[\tau]>0$, then there exists a unique CPTP map $\Alg{B}\xrightarrow{\mathcal{F}}\Alg{A}$ such that $\tau=\gamma(\mathcal{F}\star\rho_{\Alg{B}})$ if and only if 
$
{\big((\id_{\Alg{A}}\otimes (\mathcal{T}'\circ\mathcal{D}')\big)(\tau)\ge0},
$
where $\mathcal{D}'$ denotes the generalized dephasing channel analogous to $\mathcal{D}$ but with respect to a dephasing basis of eigenvectors for $\rho_{\Alg{B}}$. 


\bprf[Proof of Theorem~\ref{thm:TempCompat}]
First note that \emph{any} quantum state $\tau\in\Alg{A}\otimes\Alg{B}$ can be decomposed as a mixture of product states (also called factorizable states), that is, 
\be
\label{eq:stateasaffinecomboofsep}
\tau=\sum_{\theta} t_{\theta} \rho_{\Alg{A};\theta} \otimes \rho_{\Alg{B};\theta}
\ee
with the numbers $\{t_{\theta}\}$ defining a \emph{quasiprobability}, i.e.,  $\sum_{\theta} t_\theta=1$ (see Ref.~\cite{sanpera1998local,sperling2009representation} and Theorem 1 of Ref.~\cite{vidal1999robustness}). Using this decomposition together with the proof of Theorem~\ref{thm:tempcompSB} leads us to conclude that the expression in~\eqref{eq:choichannel}
is still the Choi matrix of the HPTP map $\mathcal{E}$ satisfying $\tau=\mathcal{E}\star\rho_{\Alg{A}}$. Therefore, we can relate the Choi matrix of $\mathcal{E}$ to $\tau$ by a similar calculation to the one appearing in the proof of Theorem~\ref{thm:tempcompSB}. Namely,
\begin{align}
    \Choi&[\mathcal{E}]= \sum_{\theta} t_{\theta}\left(\omega\odot\rho_{\Alg{A};\theta}^{T}\right)\otimes\rho_{\Alg{B};\theta} \nonumber\\
    &=\sum_{\theta}t_{\theta}\sum_{i,j}\frac{2\<i|\rho_{\Alg{A};\theta}^{T}|j\>}{p_{i}+p_{j}}|i\>\<j|\otimes\rho_{\Alg{B};\theta}\nonumber \\
    &= \sum_{\theta}t_{\theta}\sum_{i,j}\frac{2\sqrt{p_{i}p_{j}}}{p_{i}+p_{j}}\<i|\rho_{\Alg{A};\theta}^{T}|j\> \left(\rho_{\Alg{A}}^{-\frac{1}{2}}|i\>\<j|\rho_{\Alg{A}}^{-\frac{1}{2}}\right)\otimes\rho_{\Alg{B};\theta} \nonumber \\
    &=\big((\Ad_{\rho_{\Alg{A}}^{-\frac{1}{2}}}\circ\mathcal{D})\otimes\id_{\Alg{B}}\big)\big(\tau^{T_{\Alg{A}}}\big) \label{eqn:ChoiEAdDtau}
\end{align} 
In this calculation, we let the transpose be defined with respect to the chosen orthonormal eigenbasis of $\rho_{\Alg{A}}$. The second equality then follows from~\eqref{eqn:ChoiEseparable}. The third equality follows by multiplying the inside of the sum by $1=\frac{\sqrt{p_{i}p_{j}}}{\sqrt{p_{i}p_{j}}}$ and the fact that $\rho_\Alg{A}^{-\frac{1}{2}} = \sum_i \frac{1}{\sqrt{p_i}} |i\>\<i|$. In the fourth equality, we used the definition of $\mathcal{D}$ and the notation $\Ad_{V}(W):=VWV^{\dag}$ for all square matrices $V$ and $W$. Thus, $\mathcal{E}$ is CP if and only if the expression in~\eqref{eqn:ChoiEAdDtau} is positive by Choi's theorem~\cite{Ch75}. Now, since $\Ad_{\rho_{\Alg{A}}^{\frac{1}{2}}}$ is simultaneously CP and the inverse of $\Ad_{\rho_{\Alg{A}}^{-\frac{1}{2}}}$, applying $\Ad_{\rho_{\Alg{A}}^{\frac{1}{2}}}\otimes\id_{\Alg{B}}$ yields the condition that $\mathcal{E}$ is CP if and only if $(\mathcal{D}\otimes\id_{\Alg{B}})(\tau^{T_{\Alg{A}}})\ge0$, where $T_{\Alg{A}}$ denotes the transpose with respect to the eigenbasis of $\rho_{\Alg{A}}$. Since $\mathcal{T}\circ\mathcal{D}=\mathcal{D}\circ\mathcal{T}$ for such a transpose, this gives~\eqref{eqn:decoheredPPT} when $\mathcal{T}$ is the transpose defined with respect to the orthonormal basis used for $\rho_{\Alg{A}}$. 

Therefore, it remains to show that condition~\eqref{eqn:decoheredPPT} holds with the partial transpose defined with respect to any orthonormal basis. To do this, let $\tilde{\mathcal{T}}:\Alg{A}\to\Alg{A}$ denote another transpose map, defined with respect to some other orthonormal basis. Then $\tilde{\mathcal{T}}\circ\mathcal{T}$ is CP. To see this, we calculate the Choi matrix. First, we use the Choi matrix defined with respect to the basis $\{|i\>\}$ used for $\mathcal{T}$ to arrive at 
\begin{align}
\Choi[\tilde{\mathcal{T}}\circ\mathcal{T}]&=\sum_{i,j}|i\>\<j|\otimes (\tilde{\mathcal{T}}\circ\mathcal{T})\big(|i\>\<j|\big) \nonumber \\
&=\sum_{i,j}|i\>\<j|\otimes \tilde{\mathcal{T}}\big(|j\>\<i|\big) \nonumber \\
&=\sum_{\tilde{i},\tilde{j}}|\tilde{i}\>\<\tilde{j}|\otimes \tilde{\mathcal{T}}\big(|\tilde{j}\>\<\tilde{i}|\big) \nonumber \\
&=\sum_{\tilde{i},\tilde{j}}|\tilde{i}\>\<\tilde{j}|\otimes |\tilde{i}\>\<\tilde{j}|=\tilde{\Choi}[\id_{\Alg{A}}]
\end{align}
In the third equality, we used the fact that the Jamio{\l}kowski matrix is independent of the basis chosen, and we therefore used the orthonormal basis $|\tilde{i}\>$ with respect to which $\tilde{\mathcal{T}}$ is defined. The last equality follows from the fact that the resulting expression is the Choi matrix of the identity channel using the $\{|\tilde{i}\>\}$ basis. Thus, $\tilde{\mathcal{T}}\circ\mathcal{T}$ is CP by Choi's theorem. 
Now, using the fact that $\tilde{\mathcal{T}}\circ\mathcal{T}$ is CP and also using the fact that $\mathcal{T}^{2}=\id_{\Alg{A}}$, we can apply the channel $(\tilde{\mathcal{T}}\circ\mathcal{T})\otimes\id_{\Alg{B}}$ to the positive element $\big((\mathcal{T}\circ\mathcal{D})\otimes\id_{\Alg{B}}\big)(\tau)$ to yield $\big((\tilde{\mathcal{T}}\circ\mathcal{D})\otimes\id_{\Alg{B}}\big)(\tau)\ge0$. 
\eprf

It is worth noting that condition~\eqref{eqn:decoheredPPT} is closely related to the positive partial transpose (PPT) criterion~\cite{choi1982positive,peres1996separability,horodecki1996necessary}. Recall, the \define{PPT criterion} states that if a bipartite density matrix $\tau\in\Alg{A}\otimes\Alg{B}$ is separable, then $(\mathcal{T}\otimes\id_{\Alg{B}})(\tau)\ge0$. More generally, any bipartite density matrix $\tau\in\Alg{A}\otimes\Alg{B}$ satisfying $(\mathcal{T}\otimes\id_{\Alg{B}})(\tau)\ge0$ is called a \define{PPT state}. Note that this condition is independent of the orthonormal basis used to define the transpose, 
and one could instead use a partial transpose on system $\Alg{B}$ in the sense that $(\mathcal{T}\otimes\id_{\Alg{B}})(\tau)\ge0$ is equivalent to $(\id_{\Alg{A}}\otimes\mathcal{T}')(\tau)\ge0$ for the transpose $\mathcal{T}'$ on $\Alg{B}$ with respect to some orthonormal basis. The contrapositive of the PPT criterion therefore reads that if $\tau$ is a bipartite state and $(\mathcal{T}\otimes\id_{\Alg{B}})(\tau)$ is not positive, then $\tau$ is entangled. From this perspective, we can view Theorem~\ref{thm:TempCompat} as a temporal analogue of the PPT criterion valid in all dimensions. Namely, if $\tau$ is a bipartite state, then the partial transpose of $\tau$ need not be a state, but if it is a state after applying a particular decoherence channel on one factor, then $\tau$ is temporally compatible. 

We have the following general fact relating PPT states and temporal compatibility as a consequence of Theorem~\ref{thm:TempCompat}.

\bc
\label{cor:PPTimpliestempcompat}
Let $\tau\in\Alg{A}\otimes\Alg{B}$ be a PPT state, i.e., $\tau$ is a density matrix and $(\mathcal{T}\otimes\id_{\Alg{B}})(\tau)\ge0$, where the transpose map $\mathcal{T}$ is defined with respect to any orthonormal basis on $\Alg{H}_{\Alg{A}}$. Then, $\tau$ is temporally compatible in the sense that there exists a CPTP map $\Alg{A}\xrightarrow{\mathcal{E}}\Alg{B}$ such that $\tau=\mathcal{E}\star\rho_{\Alg{A}}$, where $\rho_{\Alg{A}}=\Tr_{\Alg{B}}[\tau]$.
\ec

\begin{proof}
As stated before Corollary~\ref{cor:PPTimpliestempcompat}, a state is PPT with the transpose defined with respect to one orthonormal basis if and only if it is PPT with respect to any other orthonormal basis. Therefore, we will use an orthonormal eigenbasis of $\rho_{\Alg{A}}$ for the transpose operation on $\Alg{A}$. 
Theorem~\ref{thm:TempCompat}, specifically the alternative expression~\eqref{eqn:temporalPPT}, immediately implies the claim since $\tau^{T_{\Alg{A}}}\ge0$ by assumption and because $\mathcal{D}$ is CP. 
\end{proof}

Since every separable state is PPT, Corollary~\ref{cor:PPTimpliestempcompat} gives an alternative proof of Theorem~\ref{thm:tempcompSB}.

\section{Discussion}
\label{sec:discussion}

In this work, we addressed the problem of determining whether a given bipartite density matrix is compatible with a direct causal influence by a temporal channel. Since inferring or ruling out causal structure depends on the measurements performed, we focused on the case of measurements involving light-touch observables, which include all Pauli observables, and their associated L\"uders instruments. Specifically, in Theorem~\ref{thm:TempCompat}, we gave necessary and sufficient conditions for such a temporal channel to exist. As a special case, in Theorem~\ref{thm:tempcompSB} we showed that separable density matrices always admit such a direct causal influence. The case of more general bipartite hermitian matrices giving rise to physically realizable expectation values is given in Appendix~\ref{app:generaltemporalcompatibility}. 

An interesting aspect of the necessary and sufficient conditions in Theorem~\ref{thm:TempCompat} guaranteeing when a bipartite state also admits a plausible direct causal influence explanation is the appearance of a condition closely related to the PPT criterion combined with a dephasing channel. Moreover, in the case of separable bipartite states, the dephasing channel becomes a crucial component of the temporal channel compatible with the correlations. Namely, the temporal channel is the composite of a dephasing channel combined with a pretty good measure-and-prepare protocol, thus providing the temporal channel with a concrete interpretation. It is worth mentioning that the two temporal channels going between the marginals are \emph{not} Petz recovery maps of each other, but rather, they are \emph{Bayesian inverses} of each other, where the notion of Bayesian inverse is that which appears in Ref.~\cite{FuPa22a}. Equation~\eqref{eq:DFED} reveals that the Petz recovery map and the Bayesian inverse differ by a dephasing channel, which seems consistent with recent observations made in Ref.~\cite{PaFu24TSC} for a different example involving an amplitude-damping channel. This illustrates that the Petz recovery map, although often thought to be a natural candidate of Bayesian inversion and retrodiction~\cite{PaBu22,LiWi18}, is not necessarily the only retrodiction map to use when describing quantum spatiotemporal correlations. 

Our work leads to several natural questions within the general program of determining plausible causal structures compatible with a given set of expectation values. We mention five such questions. (1) Here, we showed what states admit a plausible causal explanation. Dually, what channels induce a positive semidefinite state over time? (2) What is the geometry of the space of bipartite hermitian matrices $\tau\in\Alg{A}\otimes\Alg{B}$, with $\Tr_{\Alg{B}}[\tau]$ and $\Tr_{\Alg{A}}[\tau]$ density matrices, that are temporally compatible? By our work here, we know this space contains the convex space of separable density matrices (in fact, all PPT states), but we also know by earlier work that this space is neither convex nor does it contain all bipartite density matrices~\cite{SNREG23}. (3) To better understand compatible causal structures, it is important to investigate the case of multipartite distributions as well as more general causal structures~\cite{liu2023quantum,wolfe2021qinflation,ABHLS17,FrKl23,Fu23,BaLoOr19,LieFull2024,KAPW24}. (4) Since our results were concerned with the spatiotemporal correlations obtained from observables that have at most two eigenvalues (termed light-touch observables) with their corresponding L\"uders projective measurements, it is important to investigate extensions of our results to arbitrary positive operator-valued measures and quantum instruments~\cite{Kr83,DaLe70,Ozawa1984,BLM1996}. In fact, recent work in Ref.~\cite{LiKw25} suggests that temporal compatibility in such a general measurement scheme can be approached through weak measurement and interferometry.
(5) Since our work has been formulated in the context of finite-dimensional quantum systems, a natural next step would be to analyze how this work extends to infinite-dimensional quantum systems, where measurement must take account of both quantum field theory and general relativity~\cite{FewsterVerch2020quantum}.

\vspace{3mm}
\noindent
\textbf{Acknowledgments.}
We thank Clive Cenxin Aw, Rudrajit Banerjee, Francesco Buscemi, James Fullwood, Mile Gu, Xueyuan Hu, Felix Leditzky, Seok~Hyung Lie, Valerio Scarani, and Elie Wolfe for discussions. We have also benefited from discussions that took place during the Horizons of Quantum Information II conference in Hainan University. 
This project began while AJP was in the Graduate School of Informatics at Nagoya University and was visiting the Department of Mathematics at MIT. 
This work was partially supported by MEXT-JSPS Grant-in-Aid for Transformative Research Areas (A) ``Extreme Universe'', No.\ 21H05183. MS also thanks the National Research Foundation, Singapore through the National Quantum Office, hosted in A*STAR, under its Centre for Quantum Technologies Funding Initiative (S24Q2d0009); and the Ministry of Education, Singapore, under the Tier 2 grant ``Bayesian approach to irreversibility'' (Grant No.~MOE-T2EP50123-0002). Figure~\ref{fig:EequalsGD} was made using Mathematica~\cite{Mathematica2024}. 

\vspace{3mm}
\noindent
\textbf{Conflict of Interest.} 
The authors have no conflicts to disclose.

\vspace{3mm}
\noindent
\textbf{Data Availability Statement.}
Data sharing is not applicable to this article as no new data were created or analyzed in this study.

\appendix

\section{Well-definedness of Dephasing channel}
\label{app:Dwelldefined}

In this appendix, we will prove that $\mathcal{D}$ from~\eqref{eq:dephasing} does not depend on the choice of eigenbasis for $\rho_{\Alg{A}}$. We temporarily change notation and use the index $i$ to represent a \emph{distinct} eigenvalue $p_{i}$ of $\rho_{\Alg{A}}$. Namely, for each eigenvalue $p_{i}$ of $\rho_{\Alg{A}}$, let $\{|i,\alpha_{i}\>\}$ be an orthonormal basis for the eigenspace corresponding to $p_{i}$ with multiplicity $m_{i}$ so that the index $\alpha_{i}$ runs from $1$ to $m_{i}$ and $\rho_{\Alg{A}}=\sum_{i}\sum_{\alpha_i}p_{i}|i,\alpha_{i}\>\<i,\alpha_{i}|$. Then 
\begin{align}
\mathcal{D}(A)&=\sum_{i,j}\sum_{\alpha_{i}=1}^{m_{i}}\sum_{\alpha_{j}=1}^{m_{j}}\frac{2\sqrt{p_{i}p_{j}}}{p_{i}+p_{j}}|i,\alpha_{i}\>\<i,\alpha_{i}|A|j,\alpha_{j}\>\<j,\alpha_{j}| \nonumber \\
&=\sum_{i,j}\frac{2\sqrt{p_{i}p_{j}}}{p_{i}+p_{j}}\sum_{\alpha_{i}=1}^{m_{i}}\sum_{\alpha_{j}=1}^{m_{j}}|i,\alpha_{i}\>\<i,\alpha_{i}|A|j,\alpha_{j}\>\<j,\alpha_{j}| \nonumber \\
&=\sum_{i,j}\frac{2\sqrt{p_{i}p_{j}}}{p_{i}+p_{j}} P_{i}A P_{j},
\label{eqn:dephasingwelldefined}
\end{align}
where now $\rho_{\Alg{A}}=\sum_{i}p_{i}P_{i}$ is a spectral decomposition of $\rho_{\Alg{A}}$ with $P_{i}=\sum_{\alpha_{i}}|i,\alpha_{i}\>\<i,\alpha_{i}|$ the projection operator onto the eigenspace associated with eigenvalue $p_{i}$. Expression~\eqref{eqn:dephasingwelldefined} therefore provides a manifestly basis independent expression for $\mathcal{D}$, which proves that it is well-defined.

\section{Perfect state discrimination of orthogonal mixed states}
\label{app:statediscrim}

In this appendix, we explain the details regarding state discrimination from Section~\ref{sec:statediscrim}. For notation, we set $\delta_{\theta\vartheta}$ to be the Kronecker delta, which is $1$ when $\theta=\vartheta$ and $0$ otherwise, with $\theta,\vartheta\in\Theta$ and $\Theta$ some finite set. We also recall that a collection of states $\{\rho_{\Alg{A};\theta}\}$ in $\Alg{A}$ is \emph{orthogonal} iff $\Tr[\rho_{\Alg{A};\theta}\rho_{\Alg{A};\vartheta}]=0$ whenever $\theta\ne\vartheta$. 

\begin{definition}
\label{defn:perfectlydistinguishable}
A nonempty ensemble $\{\rho_{\Alg{A};\theta},t_{\theta}\}_{\theta\in\Theta}$, consisting of states $\rho_{\Alg{A};\theta}\in\Alg{A}$ and corresponding probabilities $t_{\theta}>0$, is \define{perfectly distinguishable} iff there exists a POVM $\{E'_{\phi}\}_{\phi\in\Phi}$, with $E'_{\phi}\in\Alg{A}$, and a surjective function $f:\Phi\to\Theta$ such that 
\be
\label{eqn:pdist}
\Tr[\rho_{\Alg{A};\theta}E'_{\phi}] \,t_{\theta}=\delta_{\theta f(\phi)}\,\Tr[\rho_{\Alg{A}}E'_{\phi}]
\ee
for all $\theta\in\Theta$ and $\phi\in\Phi$, where $\rho_{\Alg{A}}=\sum_{\theta}t_{\theta}\rho_{\Alg{A};\theta}$ is the average state of the ensemble.
\end{definition}

Condition~\eqref{eqn:pdist} comes from demanding that the conditional probability $\mathbb{P}(\theta\,|\,\phi)$ of deducing state $\rho_{\Alg{A};\theta}$ having measured $E'_{\phi}$ is $1$ if and only if $f(\phi)=\theta$ and is $0$ otherwise, i.e., $\mathbb{P}(\theta\,|\,\phi)=\delta_{\theta f(\phi)}$, so that an occurrence of $1$ for some $\phi$ can be used to unambiguously identify $\theta$ through $\theta=f(\phi)$. Thus, condition~\eqref{eqn:pdist} is just a rearrangement of the classical Bayes' rule
\be
\mathbb{P}(\phi\,|\,\theta)\;\mathbb{P}(\theta)=\mathbb{P}(\theta\,|\,\phi)\;\mathbb{P}(\phi),
\ee
with $\mathbb{P}(\phi\,|\,\theta)=\Tr[\rho_{\Alg{A};\theta}E'_{\phi}]$, $\mathbb{P}(\theta)=t_{\theta}$, and $\mathbb{P}(\phi)=\Tr[\rho_{\Alg{A}}E'_{\phi}]$. 

\begin{proposition}
An ensemble $\{\rho_{\Alg{A};\theta},t_{\theta}\}_{\theta\in\Theta}$ is orthogonal if and only if it is perfectly distinguishable.
\end{proposition}

\begin{proof}
For one direction, suppose that $\{\rho_{\Alg{A};\theta}\}$ is orthogonal. For each $\theta$, set $E'_{\theta}$ to be the orthogonal projector onto the support of $\rho_{\Alg{A};\theta}$. If $\sum_{\theta\in\Theta}E'_{\theta}$ is not equal to $1_{\Alg{A}}$, set $E'_{0}:=1_{\Alg{A}}-\sum_{\theta\in\Theta}E'_{\theta}$. Then take $\Phi:=\Theta\coprod\{0\}$, and set $f:\Phi\to\Theta$ to be the function given by $f(\theta)=\theta$ whenever $\theta\in\Theta$ and $f(0)$ to be any element of $\Theta$. Then, this POVM satisfies the conditions of Definition~\ref{defn:perfectlydistinguishable}. 

For the converse direction, suppose $\{\rho_{\Alg{A};\theta}\}$ is perfectly distinguishable via some POVM $\{E'_{\phi}\}$. 
Set 
\begin{equation}
F_{\theta}:=\sum_{\phi\in f^{-1}(\{\theta\})}E'_{\phi}
\end{equation}
so that $\{F_{\theta}\}_{\theta\in\Theta}$ is a POVM and 
\begin{align}
\Tr[\rho_{\Alg{A};\theta}F_{\vartheta}]&=\sum_{\phi\in f^{-1}(\{\vartheta\})}\Tr[\rho_{\Alg{A};\theta}E'_{\phi}] \nonumber \\
&=\sum_{\phi\in f^{-1}(\{\vartheta\})}\frac{1}{t_{\theta}}\delta_{\theta f(\phi)}\Tr[\rho_{\Alg{A}}E'_{\phi}] \nonumber \\
&=0
\end{align}
whenever $\theta\ne\vartheta$.
By using cyclicity of the trace, this gives
\begin{equation}
0=\Tr[\rho_{\Alg{A};\theta}F_{\vartheta}]=\Tr\left[\rho_{\Alg{A};\theta}^{\frac{1}{2}}{F}_{\vartheta}^{\frac{1}{2}}{F}_{\vartheta}^{\frac{1}{2}}\rho_{\Alg{A};\theta}^{\frac{1}{2}}\right],
\end{equation}
which implies ${F}_{\vartheta}^{\frac{1}{2}}\rho_{\Alg{A};\theta}^{\frac{1}{2}}=0$ (since $\Tr[A^{\dag}A]=0$ implies $A=0$), and therefore $F_{\vartheta}\rho_{\Alg{A};\theta}=0$ when $\theta\ne\vartheta$. Similarly, $\rho_{\Alg{A};\theta}F_{\vartheta}=0$ when $\theta\ne\vartheta$. Now, we suppose, to the contrary, that $\theta,\vartheta$ is a distinct pair with $\rho_{\Alg{A};\theta}$ not orthogonal to $\rho_{\Alg{A};\vartheta}$, i.e., $\Tr[\rho_{\Alg{A};\theta}\rho_{\Alg{A};\vartheta}]\ne0$. From this nonorthogonality, it follows that $\Tr[\rho_{\Alg{A};\theta}\rho_{\Alg{A};\vartheta}]$ is strictly positive. On the other hand, we also have
\begin{align}
\Tr[\rho_{\Alg{A};\theta}\rho_{\Alg{A};\vartheta}]&=\sum_{\nu\in\Theta}\Tr[\rho_{\Alg{A};\theta}F_{\nu}\rho_{\Alg{A};\vartheta}] \nonumber \\
&=\Tr[\rho_{\Alg{A};\theta}F_{\vartheta}\rho_{\Alg{A};\vartheta}]=0,
\end{align}
where the first equality follows from $\sum_{\nu}F_{\nu}=1_{\Alg{A}}$ and the latter two equalities follow from $F_{\nu}\rho_{\Alg{A};\vartheta}=0$ when $\nu\ne\vartheta$ and $\rho_{\Alg{A};\theta}F_{\vartheta}=0$ because $\theta\ne\vartheta$. 
This gives a contradiction. 
\end{proof}

\section{Criterion for temporal compatibility}
\label{app:generaltemporalcompatibility}

Theorem~\ref{thm:TempCompat} assumed that $\tau\in\Alg{A}\otimes\Alg{B}$ is a density matrix, which is more in line with the standard classical causal compatibility problem that asks what causal structures are compatible with a given joint probability distribution. However, since there are situations in quantum theory where the outcomes of measurements depend on the order in which they are measured, these expectation values cannot always be represented by density matrices~\cite{FuPa24a,LeSp13}. In the case of pseudo-density matrices, they are nevertheless represented by hermitian operators whose marginals are density matrices. As such, one might wonder if Theorem~\ref{thm:TempCompat} can be generalized to the case where $\tau$ is a hermitian matrix such that its two marginals $\Tr_{\Alg{B}}[\tau]$ and $\Tr_{\Alg{A}}[\tau]$ are states. 

Fortunately, our theorem immediately generalizes to this case, with the same condition~\eqref{eqn:decoheredPPT}. This follows from the fact that every hermitian matrix $\tau\in\Alg{A}\otimes\Alg{B}$ such that $\Tr[\tau]=1$ can also be expressed as~\eqref{eq:stateasaffinecomboofsep}, i.e., an affine mixture of separable states. This fact follows from the result for separable density matrices and the fact that every indefinite trace 1 hermitian matrix can be expressed as an affine mixture of density matrices via
\begin{equation}
\label{eqn:decomptau}
\tau=\Tr[\tau^{+}] \left(\frac{\tau^{+}}{\Tr[\tau^{+}]}\right)-\Tr[\tau^{-}] \left(\frac{\tau^{-}}{\Tr[\tau^{-}]}\right),
\end{equation}
where $\tau=\tau^{+}-\tau^{-}$ is the decomposition of a hermitian matrix as a difference of positive matrices (if $\Tr[\tau^{-}]=0$, then $\tau=\tau^{+}$ is a density matrix). Applying Theorem 1 of Ref.~\cite{vidal1999robustness} to $\frac{\tau^{\pm}}{\Tr[\tau^{\pm}]}$ and using the fact that $1=\Tr[\tau]=\Tr[\tau^{+}]-\Tr[\tau^{-}]$ shows that $\tau$ can be expressed as an affine mixture of separable density matrices. 

This result has two immediate corollaries worth mentioning. The first is that if $\rho_{\Alg{A}\Alg{B}}$ is a bipartite density matrix, then $\tau:=\rho_{\Alg{A}\Alg{B}}^{T_{\Alg{B}}}$ is hermitian and satisfies~\eqref{eqn:decoheredPPT} since $\big((\mathcal{T}\circ\mathcal{D})\otimes\id_{\Alg{B}}\big)(\tau)=\big((\mathcal{T}\otimes\mathcal{T}')\circ(\mathcal{D}\otimes\id_{\Alg{B}})\big)(\rho_{\Alg{A}\Alg{B}})\ge0$. Hence, this provides an alternative proof of the main result of Ref.~\cite{FuLi25}, which states that $\tau=\mathcal{E}\star\rho_{\Alg{A}}$ for some channel $\mathcal{E}$. The second is related to Bayesian inverses defined as follows~\cite{FuPa22a}
(alternative, but equivalent, formulations of such a Bayesian inverse have also appeared in Ref.~\cite{Ts22} as generalized conditional expectations and in Ref.~\cite{SASDS23} from an information geometry perspective). 

\bd
Given a process $(\mathcal{E},\rho_{\Alg{A}})\in\mathscr{P}(\Alg{A},\Alg{B})$, a \define{Bayesian inverse} of $(\mathcal{E},\rho_{\Alg{A}})$ is a CPTP map $\Alg{B}\xrightarrow{\mathcal{E}^{\star}_{\rho}}\Alg{A}$ such that $\mathcal{E}\star\rho_{\Alg{A}}=\gamma(\mathcal{E}^{\star}_{\rho}\star\rho_{\Alg{B}})$, where $\rho_{\Alg{B}}:=\mathcal{E}(\rho_{\Alg{A}})$. 
\ed

The equation $\mathcal{E}\star\rho_{\Alg{A}}=\gamma(\mathcal{E}^{\star}_{\rho}\star\rho_{\Alg{B}})$ in this definition is analogous (and specializes) to the Bayes' rule of classical probability theory, which reads $\mathbb{P}(y|x)\mathbb{P}(x)=\mathbb{P}(x|y)\mathbb{P}(y)$, where the quantities are probabilities and conditional probabilities~\cite{FuPa22a}. 

\bc
A process $(\mathcal{E},\rho_{\Alg{A}})\in\mathscr{P}(\Alg{A},\Alg{B})$ admits a Bayesian inverse if and only if 
\be
\big(\id_{\Alg{A}}\otimes(\mathcal{T}'\circ\mathcal{D}')\big)(\tau)\ge0,
\ee
where $\tau=\mathcal{E}\star\rho_{\Alg{A}}$. 
\ec

\section{Non-faithful marginal states}
\label{sec:nonfaithful}

In Theorem~\ref{thm:tempcompSB}, Proposition~\ref{prop:HPTPconditional}, Theorem~\ref{thm:splittemporalmap}, and related results, it was assumed that $\rho_{\Alg{A}}$ was strictly positive. This assumption can be dropped and similar results can be obtained, which are detailed as follows. 

In Theorem~\ref{thm:tempcompSB}, $\rho_{\Alg{A}}\ge0$ still guarantees the existence of $\mathcal{E}$, though it need not be unique. The proof follows similar steps, except that $\mathcal{E}\big(|i\>\<j|\big)$ is given by~\eqref{eq:qconditionalmap} only for $i,j$ satisfying $p_{i}+p_{j}>0$, i.e., whenever $p_{i}>0$ or $p_{j}>0$. Moreover, note that 
\be
\label{eq:ptrtauij}
\Tr_{\Alg{A}}\Big[\tau\big(|i\>\<j|\otimes1_{\Alg{B}}\big)\Big]=0
\ee
whenever $p_{i}=0$ or $p_{j}=0$. This important fact follows from the following lemma. 

\blem
\label{lem:kerneltensor}
Let $\Alg{A}=\matr_{m}$ and $\Alg{B}=\matr_{n}$ with $m,n$ positive integers. Given a positive semidefinite matrix $\tau\in\Alg{A}\otimes\Alg{B}$, set $\rho_{\Alg{A}}:=\Tr_{\Alg{B}}[\tau]$ and $\rho_{\Alg{B}}:=\Tr_{\Alg{A}}[\tau]$. Then 
\be
\label{eq:kernelinclusion}
\ker(\tau)\supseteq
\ker(\rho_{\Alg{A}})\otimes\C^{n}+\C^{m}\otimes\ker(\rho_{\Alg{B}}),
\ee
where $\ker(\eta)=\big\{v\in\C^{k}\;:\; \eta v=0\big\}$ denotes the kernel of the matrix $\eta\in\matr_{k}$, and $k$ is a positive integer.
\elem

\bprf
Let $|i\>\in\ker(\rho_{\Alg{A}})$ and let $\{|\psi_{l}\>\}$ be an orthonormal basis of vectors in $\C^{n}$, so that the index $l$ runs from $1$ to $n$. Then, since $\tau$ is a positive semidefinite operator, 
\begin{equation}
0\le\big(\<i|\otimes\<\psi_{l}|\big)\tau\big(|i\>\otimes|\psi_{l}\>\big).
\end{equation}
This implies 
\begin{align}
0&\le\sum_{l=1}^{n}\big(\<i|\otimes\<\psi_{l}|\big)\tau\big(|i\>\otimes|\psi_{l}\>\big)\nonumber\\
&=\Tr\left[\tau\left(|i\>\<i|\otimes\sum_{l=1}^{n}|\psi_{l}\>\<\psi_{l}|\right)\right] \nonumber \\
&=\big\<i\big|\Tr_{\Alg{B}}[\tau]\big|i\big\>
=\<i|\rho_{\Alg{A}}|i\>=0. 
\end{align}
Hence, $\big(\<i|\otimes\<\psi_{l}|\big)\tau\big(|i\>\otimes|\psi_{l}\>\big)=0$. Therefore, $\tau\big(|i\>\otimes|\psi_{l}\>\big)=0$, which proves that $|i\>\otimes|\psi_{l}\>\in\ker(\tau)$. This proves that $\ker(\rho_{\Alg{A}})\otimes\C^{n}\subseteq\ker(\tau)$. Similarly, $\C^{m}\otimes\ker(\rho_{\Alg{B}})\subseteq\ker(\tau)$. These two inclusions imply~\eqref{eq:kernelinclusion}. 
\eprf

Lemma~\ref{lem:kerneltensor} implies that $\tau \big(|i\>\<j|\otimes1_{\Alg{B}}\big)=0$ whenever $p_{i}=0$, which implies~\eqref{eq:ptrtauij} whenever $p_{i}=0$. Moreover, due to the cyclicity of the partial trace $\Tr_{\Alg{A}}$ in the first coordinate (i.e., $\Alg{A}$), and the fact that $\tau^{\dag}=\tau$, we also conclude~\eqref{eq:ptrtauij} whenever $p_{j}=0$.  
Therefore, one obtains $\mathcal{E}\big(|i\>\<j|\big)=0$ if $p_{i}=0$ or $p_{j}=0$ and $i\ne j$. 

Meanwhile, one can set $\mathcal{E}\big(|i\>\<i|\big)=\frac{1_{\Alg{B}}}{n}$ if $p_{i}=0$ (other specifications are possible---the one given here is sufficient for our purposes to illustrate existence). Then, $\mathcal{E}$ is defined by linear extension. Namely, putting all this together, $\mathcal{E}$ is explicitly given on $A\in\Alg{A}$ by 
\be
\label{eqn:temporalmapnonfaithful}
\mathcal{E}(A):=\!\sum_{\substack{i,j\\p_{i},p_{j}>0}}\!\frac{2\<i|A|j\>}{p_i+p_j}\Tr_{\Alg{A}}\Big[\tau\big(|i\>\<j|\otimes1_{\Alg{B}}\big)\Big]+\Tr[P^{\perp} A]\frac{1_{\Alg{B}}}{n},
\ee
where $P$ denotes the support projection of $\rho_{\Alg{A}}$ and $P^{\perp}:=1_{\Alg{A}}-P$ is the projection onto its orthogonal complement, i.e., the kernel of $\rho_{\Alg{A}}$. The map $\mathcal{E}$ as given by~\eqref{eqn:temporalmapnonfaithful} provides a generalization of Proposition~\ref{prop:HPTPconditional}, while the map $\mathcal{E}$ is still guaranteed to satisfy the conditions of Theorem~\ref{thm:tempcompSB}, meaning, $\mathcal{E}$ is CPTP whenever $\tau$ is separable.

In Theorem~\ref{thm:splittemporalmap}, the maps $\mathcal{D}$ and $\mathcal{G}$ were defined under the assumption that $\rho_{\Alg{A}}>0$. In the more general setting with $\rho_{\Alg{A}}\ge0$, we set $\mathcal{D}$ to be 
\be
\mathcal{D}(A):=\sum_{\substack{i,j\\p_{i},p_{j}>0}}\frac{2\sqrt{p_{i}p_{j}}}{p_{i}+p_{j}}|i\>\<i|A|j\>\<j|+
\Tr[P^{\perp}A]\frac{1_{\Alg{A}}}{m}
\ee
and $\mathcal{G}$ to be 
\be
\mathcal{G}(A):=\Tr_{\Alg{A}}\Big[\tau\Big(\big(\rho_{\Alg{A}}^{-\frac{1}{2}}A\rho_{\Alg{A}}^{-\frac{1}{2}}\big)\otimes1_{\Alg{B}}\Big)\Big]+\Tr[P^{\perp}A]\frac{1_{\Alg{B}}}{n}.
\ee 
In this last expression, inverses of $\rho_{\Alg{A}}$ are assumed to be Moore--Penrose pseudoinverses~\cite{Mo1920,Pe55,Strang2022}. An alternative expression for $\mathcal{G}$ in terms of the $\{G_{\theta}\}$ from~\eqref{eq:PGMPOVM} (though using pseudoinverses) would be 
\be
\mathcal{G}(A)=\sum_{\theta}\Tr[G_{\theta}A]\rho_{\Alg{B};\theta}+\Tr[P^{\perp} A]\frac{1_{\Alg{B}}}{n}.
\ee
Since the collection $\{G_{\theta}\}$ alone does not define a POVM due to the fact that $\sum_{\theta}G_{\theta}=P$, one adjoins the operator $G_{0}=P^{\perp}$ to get a POVM so that $\mathcal{G}$ still retains the interpretation of a measure-and-prepare channel. 

\br
Mathematically, the fact that $\mathcal{E}\big(|i\>\<j|\big)=0$ if $p_{i}=0$ or $p_{j}=0$ is a consequence of the form of the canonical state over time. If we had used $(\rho\otimes1_{\Alg{B}})\Jamio[\mathcal{E}]$ or $\Jamio[\mathcal{E}](\rho\otimes1_{\Alg{B}})$ instead, then $\mathcal{E}\big(|i\>\<j|\big)$ might not vanish if $p_{i}=0$ or $p_{j}=0$, and we would have to resort to a more subtle analysis along the lines of Refs.~\cite{PaRuBayes,GPRR21,PaRu19}.
\er


\addcontentsline{toc}{section}{\numberline{}Bibliography}
\bibliography{references}

\end{document}